%% file: main.tex
\documentclass[nonacm,sigconf]{acmart}

\AtBeginDocument{%
  }

\usepackage{amsmath}
\usepackage{caption}
\usepackage{multirow}
\usepackage[linesnumbered,lined,boxed,commentsnumbered,ruled,longend, vlined]{algorithm2e}

\newcommand{\T}{{\mathcal T}}

\usepackage{booktabs} 
\usepackage{enumitem}
\usepackage{mathtools}
\usepackage{todonotes}
 \usepackage{nicematrix}

 \usepackage{mathtools, nccmath}

\begin{document}

%%
%% The "title" command has an optional parameter,
%% allowing the author to define a "short title" to be used in page headers.
\title{On the Efficiency of Dynamic Transaction Scheduling in Blockchain Sharding}

%%
%% The "author" command and its associated commands are used to define
%% the authors and their affiliations.
%% Of note is the shared affiliation of the first two authors, and the
%% "authornote" and "authornotemark" commands
%% used to denote shared contribution to the research.

% uncomment from here for authors name
% \author{Ben Trovato}
% \authornote{Both authors contributed equally to this research.}
% \email{trovato@corporation.com}
% \orcid{1234-5678-9012}
% \author{G.K.M. Tobin}
% \authornotemark[1]
% \email{webmaster@marysville-ohio.com}
% \affiliation{%
%   \institution{Institute for Clarity in Documentation}
%   \streetaddress{P.O. Box 1212}
%   \city{Dublin}
%   \state{Ohio}
%   \country{USA}
%   \postcode{43017-6221}
% }

\author{Ramesh Adhikari}
\orcid{0000-0002-8200-9046}
\affiliation{%
  \institution{School of Computer \& Cyber Sciences\\Augusta University}
  % \streetaddress{P.O. Box 1212}
  \city{Augusta}
  \state{Georgia}
  \country{USA}
  \postcode{30912}
  }
\email{radhikari@augusta.edu}

\author{Costas Busch}
\orcid{0000-0002-4381-4333}
\affiliation{%
  \institution{School of Computer \& Cyber Sciences\\Augusta University}
  % \streetaddress{P.O. Box 1212}
  \city{Augusta}
  \state{Georgia}
  \country{USA}
  \postcode{30912}
}
\email{kbusch@augusta.edu}

\author{Miroslav Popovic}
\orcid{0000-0001-8385-149X}
\affiliation{%
  \institution{Faculty of Technical Sciences \\ University of Novi Sad}
  % \streetaddress{P.O. Box 1212}
  \city{Novi Sad}
  \state{}
  \country{Serbia}
  \postcode{}
}
\email{miroslav.popovic@rt-rk.uns.ac.rs}

% \author{Valerie B\'eranger}
% \affiliation{%
%   \institution{Inria Paris-Rocquencourt}
%   \city{Rocquencourt}
%   \country{France}
% }

% \author{Aparna Patel}
% \affiliation{%
%  \institution{Rajiv Gandhi University}
%  \streetaddress{Rono-Hills}
%  \city{Doimukh}
%  \state{Arunachal Pradesh}
%  \country{India}}

% \author{Huifen Chan}
% \affiliation{%
%   \institution{Tsinghua University}
%   \streetaddress{30 Shuangqing Rd}
%   \city{Haidian Qu}
%   \state{Beijing Shi}
%   \country{China}}

% \author{Charles Palmer}
% \affiliation{%
%   \institution{Palmer Research Laboratories}
%   \streetaddress{8600 Datapoint Drive}
%   \city{San Antonio}
%   \state{Texas}
%   \country{USA}
%   \postcode{78229}}
% \email{cpalmer@prl.com}

% \author{John Smith}
% \affiliation{%
%   \institution{The Th{\o}rv{\"a}ld Group}
%   \streetaddress{1 Th{\o}rv{\"a}ld Circle}
%   \city{Hekla}
%   \country{Iceland}}
% \email{jsmith@affiliation.org}

% \author{Julius P. Kumquat}
% \affiliation{%
%   \institution{The Kumquat Consortium}
%   \city{New York}
%   \country{USA}}
% \email{jpkumquat@consortium.net}

% %%
% %% By default, the full list of authors will be used in the page
% %% headers. Often, this list is too long, and will overlap
% %% other information printed in the page headers. This command allows
% %% the author to define a more concise list
% %% of authors' names for this purpose.

% \renewcommand{\shortauthors}{Adhika et al.}

%%
%% The abstract is a short summary of the work to be presented in the
%% article.
\begin{abstract}
Sharding is a technique to speed up transaction processing in blockchains, where the $n$ processing nodes in the blockchain are divided into $s$ disjoint groups (shards) that can process transactions in parallel. We study dynamic scheduling problems on a shard graph $G_s$ where transactions arrive online over time and are not known in advance.
Each transaction may access at most $k$ 
shards, and we denote by $d$ the worst distance between a transaction and its accessing (destination) shards (the parameter $d$ is unknown to the shards).
To handle different values of $d$, we assume a locality sensitive decomposition of $G_s$ into clusters of shards, where every cluster has a {\em leader shard} that schedules transactions for the cluster.
We first examine the simpler case of the {\em stateless model},
where leaders 
are not aware of the current state of the transaction accounts,
and we prove a $O(d \log^2 s \cdot \min\{k, \sqrt{s}\})$ competitive ratio for latency.
We then consider 
the {\em stateful model}, 
where leader shards gather 
the current state of accounts, and we prove a $O(\log s\cdot \min\{k, \sqrt{s}\}+\log^2 s)$ competitive ratio for latency.
Each leader calculates the schedule in polynomial time for each transaction that it processes. 
We show that for any $\epsilon > 0$, approximating the optimal schedule within a $(\min\{k, \sqrt{s}\})^{1 -\epsilon}$ factor is NP-hard. Hence, our bound for the stateful model is within a poly-log factor from the best possibly achievable.  
To the best of our knowledge, this is the first work to establish
provably efficient dynamic scheduling algorithms for blockchain sharding systems.
\end{abstract}




\maketitle

\input{introduction}

\input{related-work}

\input{prelimanaries}

\input{stateless/single-leader}

\input{stateless/multi-leader}

\input{stateful/single-leader}

\input{stateful/multi-leader}

% \input{simulation}
\input{conclusion}

\begin{acks}
This paper is supported by NSF grant CNS-2131538.
\end{acks}

%% the bibliography file.
\bibliographystyle{ACM-Reference-Format}
\balance
\bibliography{references}

%%
%% If your work has an appendix, this is the place to put it.
% \appendix
% \input{appendix}
\end{document}

%% file: introduction.tex
\section{Introduction}
\label{sec:introduction}

Blockchains are known for their special features, such as fault tolerance, transparency, non-repudiation, immutability, and security, and have been used in various applications and domains~\cite{monrat2019survey}.
%, and financial institutions~~\cite{androulaki2020privacy,javaid2022review,alamsyah2024review, baum2023sok}. 
% Blockchain technology is also gaining interest from banks and financial institutions
% % \todo{This paragraph is for AFT}
% because of its transformational power in the financial industry~\cite{androulaki2020privacy,javaid2022review,alamsyah2024review, weerawarna2023emerging}. This innovation is expected to reshape financial services, revolutionize banking operations, and address dynamic industry demands without sacrificing privacy~\cite{baum2023sok,weerawarna2023emerging}. Implementing blockchain can offer competitive benefits to traditional banking systems by enhancing transaction privacy and security at a minimum cost~\cite{cucari2022impact}. Traditional financial institutions, tech firms, and startups are also coordinating with financial technology to offer affordable and easy to use financial services~\cite{canaday2017evolving}. 
%  However, the primary drawback of blockchain is its scalability and its performance, which is dependent on the throughput of the transactions~\cite{cocco2017banking,chen2020blockchain}. Thus, there is a research gap to address the scalability issues of blockchain.
% \ra{Blockchain technology is also gaining interest from Transportation Systems~\cite{9430722}, Mobile Edge Computing~\cite{9387145}, cloud-edge-end and cooperative network~\cite{yuan2022coopedge}, 
% % and distributed machine learning~\cite{du2023accelerating}
% because of its secure, deuniversal, and fault tolerance attributes~\cite{bitcoin,Byshard}.} 
 However, a drawback of blockchains is that the size of the blockchain network may impact the latency and throughput of transaction processing.
 To append a new block in a blockchain network, the participating nodes reach consensus, which is a time and energy-consuming process~\cite{adhikari2024spaastable}. Moreover, each node is required to process and store all transactions, which leads to scalability issues in the blockchain system. 
 {\em Sharding} protocols have been proposed to address the scalability and performance issues of blockchains~\cite{Elastico,Rapidchain,Byshard,adhikari2023lockless}, which divide the overall blockchain network into smaller groups of nodes called {\em shards} that allow for processing transactions in parallel. In the sharded blockchain, independent transactions are processed and committed in multiple shards concurrently, which improves the blockchain system's throughput. However, most of the existing sharding protocols~\cite{Elastico,OmniLedger,Rapidchain,adhikari2023lockless} do not provide formal analysis for the scheduling time complexity (i.e. how fast the transactions can be processed). 

We consider a blockchain system consisting of $n$ nodes, which are further divided into $s$ shards, where each shard consists of  $n/s$ nodes.
Shards are connected in a graph network $G_s$ with a diameter $D$, and each shard holds a subset of the objects (transaction accounts). We assume that transactions are distributed across the shards, and each transaction accesses at most $k$ accounts.
A transaction $T_i$ initially is in one of the shards, which is called the {\em home shard} for $T_i$. For simplicity, we consider each shard has one transaction at a time, and when that transaction gets processed (either commit or abort), a new transaction will be generated at the home shard.
Similar to other sharding systems~\cite{Byshard,adhikari2023lockless,adhikari2024spaastable}, each transaction $T_i$ is split into subtransactions,
where each subtransaction accesses an account. 
A subtransaction of $T_i$ is sent to the {\em destination shard} that holds the respective account. We assume that the maximum distance between the home shard of a transaction and the respective destination shards in $G_s$ is at most $d\leq D$.
(The parameter $d$ is not known to the system.)
% and we denote $\mathfrak{D}$ as an upper bound on the local processing time and the communication delay between any two shards. 
% In other words, the transaction and its destination shards (objects) are arbitrarily far from each other with a maximum distance of $d$. 

All home shards process transactions concurrently.
A problem occurs when {\em conflicting} transactions
try to access the same account simultaneously.
In such a case, the conflict prohibits the transactions from being committed concurrently
and forces them to serialize~\cite{adhikari2024spaastable}.
Our proposed {\em scheduling algorithms} coordinate the home shards and destination shards to process the transactions (and respective subtransactions) in a conflict-free manner in polynomial time.
Each destination shard maintains a local blockchain of the subtransactions that are sent to it.
The global blockchain can be constructed (if needed) 
by combining the 
local blockchains at the shards~\cite{adhikari2023lockless}.

We consider online dynamic transaction scheduling problem instances where transactions are not known a priori. Moreover, transactions may arrive online and continuously over time, which are generated by electronic devices or some crypto app that resides on shards. %and need to be processed.
% Each shard consists of transactions and a single account. 
% We study the shard model in different graph topologies such as General graph, Clique (Complete) graph, Hypercube, Butterfly, and g-dimensional Grid graph, and provide the analysis for execution time. 
% % % \todo[]{explain why these graph typologies are important for IoT and mobile computing? why are we doing this?}
% Considering the different graph topologies for the IoT and mobile computing is important as presented in~\cite{de2018impact,mamat2019network}.
%with upper and lower bounds.
   Our proposed schedulers determine the time step for each transaction $T_i \in \T$ to process and commit. The execution of our scheduling algorithm is partially synchronous, where communication delay is upper bounded by a system parameter.
   %(see Section~\ref{preliminaries}). 
% Similar to the work in~\cite{adhikari2024spaastable}, we consider that the duration of a round is sufficient to allow 
% the execution of the PBFT consensus algorithm~\cite{PBFT} in each shard.
% A round is also the time to send a message between shards in a unit distance.
The goal of a scheduling algorithm is to efficiently process all transactions while minimizing the total execution time (makespan). 
Unlike previous sharding approaches~\cite{Byshard,OmniLedger,Rapidchain}, our scheduling algorithms are lock-free, namely, they do not require locking mechanisms for concurrency control.
%hence, our algorithms are lock-free.

We use the notion of {\em competitive ratio}~\cite{busch2022dynamic} to determine the performance of our scheduling algorithms. The competitive ratio typically measures how well a given online algorithm performs compared to the best possible offline algorithm for a specific sequence of operations.
However, in our model, the transactions generated in the future depend on the execution history. Hence, we define the competitive ratio to capture the volatile transaction history.
% of the competitive ratio to capture the dynamic 
% We give an appropriate definition that expresses how well the dynamic schedule performs compared to would be possible in 
%It is the ratio between the total processing time of the proposed algorithm and the minimum processing time that any optimal offline algorithm could achieve~\cite{busch2022dynamic}. The aim is to minimize this ratio, ideally to $1$, but it is typically difficult to achieve.\\

\paragraph*{Contributions} 
To our knowledge, this is the first work to present provably efficient online transaction scheduling algorithms for blockchain sharding systems. We summarize our contributions as follows (also see Table~\ref{tbl:contribution-summary}):

\begin{itemize}
    \item {\bf Stateless Scheduling Model:} We first provide transaction scheduling algorithms for the stateless model, where a {\em leader shard} that is responsible for coordinating transaction execution, does not require knowledge of the current state of the accessed accounts.
    %(see Section~\ref{sec:stateless-and-statefull-model}). 
    In this model, we provide two scheduling algorithms:
    \begin{itemize}
        \item {\bf Single-Leader Scheduler:} In this scheduling algorithm, one of the shards acts as the leader and all other shards send their transaction information to this leader, which determines the global transaction schedule. Our algorithm works in a partially synchronous communication model, but for the sake of performance analysis purposes, we assume a synchronous model. Let the shard network be represented as a general graph $G_s$, where each transaction accesses at most $k$ objects (shards). The maximum distance between home shards, accessed shards, and leader is denoted with $d$. Then, the single-leader scheduler achieves an $O(d \cdot \min\{k, \sqrt{s}\})$ competitive ratio with respect to the optimal scheduler. In the special case where $G_s$ is a clique with unit distances (i.e., $d = 1$), the competitive ratio becomes $O(\min\{k, \sqrt{s}\})$.

        \item {\bf Multi-Leader Scheduler:} A drawback of the single-leader case is that the distance $d$ involves also the position of the leader. On the other hand, in the multi-leader case, $d$ only involves distances between home and respective destination shards. In this scheduler, multiple leaders process the transactions, which distribute the scheduling load among multiple shards. The multi-leader approach allows for a better adaptation to the value $d$ without requiring knowledge of $d$ and without involving distances to the leaders in the definition of $d$.  This approach uses a hierarchical clustering technique~\cite{gupta2006oblivious} to cluster the shard network, which enables the independent scheduling and commitment of transactions within different clusters. This scheduler achieves a competitive ratio of $O(d \log^2 s \cdot \min\{k, \sqrt{s}\})$.
    \end{itemize}

    \item {\bf Stateful Scheduling Model:} We next consider a stateful model where the leader shard requires knowledge of the account states. Namely, a leader shard receives the transactions from the home shards (where transactions are initially generated), and then the leader shard first gathers the current state of the accounts from their corresponding account shards before scheduling and pre-committing the transactions.
    %(see Section~\ref{sec:stateless-and-statefull-model}). 
    After receiving the state, the leader pre-commits the transactions locally and forwards the pre-committed batch to the destination shards. In this model, the single-leader scheduler achieves a competitive ratio of $O(\min\{k, \sqrt{s}\})$ and the multi-leader scheduler achieves a competitive ratio of $O(\log s\cdot \min\{k, \sqrt{s}\}+\log^2 s)$. Note that these competitive ratios do not depend on $d$ (in contrast to the stateless model), which is the benefit of the stateful approach.

    \item {\bf Approximation Hardness:}   We also show that for any $\epsilon > 0$, obtaining competitive ratio $(\min\{k, \sqrt{s}\})^{1 - \epsilon}$ is NP-hard. Hence, our bound for the stateful single-leader scheduler is asymptotically the best we can achieve in polynomial time, and the bound for the stateful multi-leader scheduler is within a poly-log factor of the best achievable.

    \item {\bf Safety and Liveness Analysis:} We formally analyze the correctness of our proposed schedulers by proving both safety and liveness for the single-leader and multi-leader algorithms.
\end{itemize}

\noindent{\bf Paper Organization:}
The rest of this paper is structured as follows. 
Section \ref{sec:related-work} provides related works. Section \ref{preliminaries} describes the preliminaries for this study and the sharding model. 
Section~\ref{sec:stateless-scheduler} presents a stateless scheduling model with single-leader and multi-leader scheduling algorithms.
In Section~\ref{sec:stateful-scheduler}, we provide the stateful single-leader and multi-leader scheduling algorithms.
Finally, we give our conclusions in Section~\ref{sec:conclusion}.
% The safety and liveness proofs of our algorithms appear in the appendix.

%% file: related-work.tex
\section{Related Work}
\label{sec:related-work}

To solve the scalability issue of blockchain, various sharding protocols \cite{Elastico,Rapidchain,Byshard,Zilliqa,Nightshade,Harmony} have been proposed. These protocols have shown promising enhancements in the transaction throughput of blockchain by processing transactions in parallel in multiple shards. However, none of these protocols have specifically explored the theoretical analysis of online transaction scheduling problems in a sharding environment.
% \todo{IoT sharding related work}
% \ra{There are many research work has been done to securely store and process the IoT and mobile computing data in blockchain\cite{queralta2021blockchain,li2021blockchain,9387145}, and blockchain sharding \cite{9324984}.  However, these works also lack the theoretical analysis for the optimal schedule.}
To process transactions in parallel in the sharding model, some research work \cite{OmniLedger,Byshard} has used two-phase locking for concurrency control. However, locks are expensive because when one process locks shared data for reading/writing, all other processes attempting to access the same data set are blocked until the lock is released, which lowers system throughput. Moreover, locks, if not handled and released properly, may cause deadlocks. Our scheduling algorithms do not use locks, as concurrency control is managed by scheduling non-conflicting transactions in parallel. In~\cite{adhikari2023lockless} the authors propose lockless blockchain sharding using multi-version concurrency control. However, they lack a performance analysis, and they do not explore the benefits of locality and optimization techniques for transaction scheduling.

\renewcommand{\arraystretch}{1.6}
\begin{table*}[t!]
\centering
\small % or \footnotesize, or even \scriptsize
\resizebox{\textwidth}{!}{%
\begin{tabular}{|l|l|l|l|l|}
   \cline{1-5} 
 \multirow{2}{*}{{  }}   & \multicolumn{2}{l|}{ { {\bf \qquad\qquad\qquad Proposed Results}}}&\multicolumn{2}{l|}{ { {\bf \qquad \qquad\qquad Related works}}}
 \\
\cline{2-5} 
   &{{\bf Stateless Model}}  &{\bf Stateful Model}  &{{\bf In~\cite{adhikari2024spaastable}}}  &{\bf In~\cite{adhikari2024fast,liu2024dynashard}} \\
 \cline{1-5}
 Problem & Dynamic Transaction & Dynamic Transaction & Dynamic Transaction & Batch Transaction\\
 \cline{1-5}
 Focus & Performance & Perofrmance & Stability & Performance\\
 \cline{1-5}
  Single Leader   &$O\big(d \cdot \min\{k, \sqrt{s}\}\big)$ &  $O\big(\min\{k, \sqrt{s}\}\big) $ &$36bd \cdot \min \{k, \lceil \sqrt{s} \rceil \}$&$O\big(kd)$\\ 
 \cline{1-5}
 Multi-Leader & $O\big( d \log^2 s \cdot \min\{k, \sqrt{s}\} \big) $&  $O\big( \log s \cdot \min\{k, \sqrt{s}\} + \log^2 s\big)$  &$2 \cdot c_1' bd \log^2 s \cdot \min \{k, \lceil \sqrt{s}\rceil\}$& $O\big(kd\cdot\log d \log s\big)$\\ 
   % & & $ +\log^2 s)\big)$  &&\\ 
 \cline{1-5}
 Com. Model & Partial-synchronous&  Partial-synchronous & Synchronous& Synchronous \\ 
 \cline{1-5}
\end{tabular}
} % End of resizebox
\caption{
Comparison of our proposed online transaction scheduling algorithm's competitive ratio with related works~\cite{adhikari2024spaastable, adhikari2024fast,liu2024dynashard}. The used notations are as follows: $s$ represents a total number of shards, $k$ denotes the maximum number of shards (objects) accessed by each transaction, $d$ denotes the worst distance between any transaction (home shard) and its accessed objects (destination shard), $b$ denotes the burstiness and $c_1'$ represents some positive constant.
(Note that the bounds in \cite{adhikari2024spaastable} are the actual transaction latencies.)}
\label{tbl:contribution-summary}
\end{table*}

In a recent work \cite{adhikari2024spaastable} (see Table~\ref{tbl:contribution-summary}), the authors provide a stability analysis of blockchain sharding considering adversarial transaction generation. Their focus is on stability, not on performance, where they want to maintain a bounded pending transaction queue size and latency. 
They consider adversarial transaction generation, where at any time interval of duration $t$, the number of generated transactions using any object is bounded by $\rho t + b$, where $\rho \leq 1$ is the transaction injection rate per unit time and $b>0$ is a burstiness injection parameter.
They consider stateless scheduling model, and for the single leader scheduler where the shards are connected in the clique graph with unit distance they provide the stable transaction rate  $\rho \leq \max\{ \frac{1}{18k}, \frac{1}{\lceil 18 \sqrt{s} \rceil} \}$, for which they show the number of pending transactions at any round is at most $4bs$ (which is the upper bound on queue size in each shard), and the latency of transactions is bounded by $36b \cdot \min \{k, \lceil \sqrt{s} \rceil \}$,
%where $\rho \leq 1$ represents the transaction injection rate per unit time and  $b > 0$ is the burstiness. 
If this single leader scheduler is considered in the general graph where the transaction and its accessing object are $d$ far away, then their latency becomes $36bd \cdot \min \{k, \lceil \sqrt{s} \rceil \}$.  Similarly, for a multi-leader scheduler, they provide a stable transaction injection rate 
    $\rho \leq \frac{1}{c_1'd \log^2 s} \cdot \max\{ \frac{1}{k}, \frac{1}{\sqrt{s}} \}$,
    where $c_1'$ is some positive constant.
    For this scheduling algorithm, they show the upper bound on queue size as $4bs$, and transaction latency as 
        $2 \cdot c_1' bd \log^2 s \cdot \min \{k, \lceil \sqrt{s}\rceil\}$. However, they consider a synchronous communication model, which is not practical in blockchain, and they also do not provide a theoretical analysis of the optimal approximation for the scheduling algorithm, and they only consider a stateless scheduling model.
All their latency bounds depend on the burstiness parameter $b$, which can be arbitrarily large, as it expresses a transaction injection burst of arbitrary size in any given time interval.
On the other hand, our system models do not depend on any burstiness parameter,
as we adopt a transaction injection model tuned for performance analysis.

In~\cite{adhikari2024fast,liu2024dynashard} (see Table~\ref{tbl:contribution-summary}), the authors presented batch scheduling algorithms (for a given set of transactions) while they did not consider dynamic transaction generation. Moreover, their provided bounds are not tight even for batch processing. Furthermore, their algorithms work on a synchronous communication model, which might not be applicable in a real-world distributed blockchain network.
The authors in~\cite{liu2024dynashard} only consider single leader algorithms and have worse performance complexity bounds than~\cite{adhikari2024fast} by a factor of $\log D$, resulting in a complexity of 
$O(kd \cdot \log D)$ whereas~\cite{adhikari2024fast} achieves $O(kd)$ approximation for batch transactions. Here, we provide efficient scheduling algorithms with theoretical analysis for dynamic transaction processing in a blockchain sharding system that works in the partially synchronous communication model.

Several works have been conducted on transaction scheduling in shared memory multi-core systems, distributed systems, and transactional memory~\cite{busch2017fast,busch2022dynamic}. In ~\cite{attiya2015directory,sharma2014distributed,sharma2015load}, the authors explored transaction scheduling in distributed transactional memory systems aimed to achieve better performance bounds with low communication costs.  In \cite{busch2017fast} they provide offline scheduling for transactional memory, where each transaction attempts to access an object, and once it obtains the object, it executes the transaction. In another work \cite{busch2022dynamic}, the authors extended their analysis from offline to online scheduling for the transactional memory in a synchronous communication model. However, these works do not address transaction scheduling problems in the context of blockchain sharding. This is because, in the transactional memory model, the considered system models assume that objects are mobile, and once a transaction obtains the object, it immediately executes the transaction. In contrast, in blockchain sharding, an object is static in a shard, and there is a confirmation scheme to confirm and commit each subtransaction consistently in the respective shard.

%% file: prelimanaries.tex
\section{Technical Preliminaries}
\label{preliminaries}

%\begin{definition}[{\sc Block and Blockchain}]
% A block is a data structure that consists of a set of transactions. The block header contains additional metadata, including the block hash, previous block hash, block sequence, etc.
% Blockchain is simply a chain of blocks, where a block points to its predecessor (parent) through its hash, which makes them immutable. 
% A blockchain is essentially implemented as a decentralized peer-to-peer ledger where the blockchain is replicated across multiple interconnected nodes.

%\end{definition}

\subsection{\bf Blockchain Sharding Model:}
 We consider a blockchain sharding model similar to ~\cite{Byshard,adhikari2023lockless,adhikari2024spaastable,adhikari2024fast}, consisting of $n$ nodes 
which are 
partitioned into $s$ shards $S_1, S_2,\dots, S_s$ such that $S_i \subseteq \{1, \ldots, n\}$, for $i \neq j$, $S_i \cap S_j = \emptyset$, $n = \sum_i |S_i|$, and $n_i = |S_i|$ denotes the number of nodes in shard $S_i$.
Let $G_s = (V,E,w)$ denote a weighed graph of shards, where $V = \{S_1, S_2,\dots, S_s\}$, the edges $E$ correspond to the connections between the shards,
and the weight function $w$ represents the distance between the shards.
The graph $G_s$ is complete, since each pair of shards can communicate directly, but the weights of the edges may be non-uniform.

Each shard maintains a local blockchain (which is part of the global blockchain) according to its local accounts and the subtransactions it receives and commits. 
We use $f_i$ to represent the number of Byzantine nodes in shard $S_i$.
To guarantee consensus on the current state of the local blockchain, we assume that every shard executes the PBFT \cite{PBFT} or a similar consensus algorithm.
In order to achieve Byzantine fault tolerance, we assume each shard $S_i$ consists of $n_i > 3 f_i$ nodes.
%, where $n_i$ is the total number of nodes in $S_i$. 

% We assume that shards communicate with each other via message passing~\cite{Byshard}, and here, we are not focusing on optimizing the message size.
% Moreover, all honest nodes in a shard agree on each message before transmission
% (e.g. running the PBFT \cite{PBFT} consensus algorithm within the shard).
We assume that shards communicate with each other via message passing~\cite{Byshard}, and here, we are not focusing on optimizing the message size.
% Moreover, all honest nodes in a shard agree on each message before transmission
% (e.g. running the PBFT \cite{PBFT} consensus algorithm within the shard).
We adopt the cluster-sending protocol described in~\cite{hellings2022fault} and Byshard~\cite{Byshard}, where shards run consensus (e.g., the PBFT \cite{PBFT} consensus algorithm within the shard) before sending a message.
    %Similar to~[2], we will clarify this by replacing it with the following description based on [Jelle Hellings and Mohammad Sadoghi, 2022]. 
    For
    communication between shards $S_1$ and $S_2$, a set $A_1 \subseteq S_1$ of $f_1+1$ nodes in $S_1$ and a set $A_2 \subseteq S_2$ of  $f_2 + 1$ nodes in $S_2$ are chosen (where $f_i$ is the number of faulty nodes in shard $S_i$). Each node in $A_1$ is instructed to broadcast the message to all nodes in $A_2$. Thus, at least one non-faulty node in $S_1$ will send the correct message value to a non-faulty node in $S_2$. (Actually, $A_1$ needs to have size $2f_1 + 1$ to distinguish the correct message).
We consider a partial-synchronous communication model, 
where sending messages for transactions to their accessing shards has a bounded delay.
% The parameter $d$ is determined by two delays:
% (i) the local consensus worst time ($\delta_{cons}$) within a shard,
% which we normalize to be at most one time unit (i.e., $\delta_{cons} = 1$); 
% and (ii) the worst network delay ($\delta_{comm}$), which can be arbitrary depending on the network delays.
% Hence, $d = \delta_{cons} + \delta_{comm}$. 

% As per ~\cite{adhikari2023lockless}, the local chains can be combined and serialized as needed to create a single global blockchain.

%\todo[inline]{CB:Say somewhere that in our algorithms our blocks are simple so that each block contains only one transaction. Our algorithms can be extended to have multiple transactions per block. \ra{updated}}

Suppose we have a set of shared accounts $\mathcal{O}$ (which we also call {\em objects}). 
Similar to previous works in \cite{Byshard,adhikari2023lockless,adhikari2024spaastable}, we assume that each shard is responsible for a specific subset of the shared objects (accounts).
To be more specific, $\mathcal{O}$ is split into disjoint subsets $\mathcal{O}_1, \ldots, \mathcal{O}_s$, where the set of accounts under the control of shard $S_i$ is represented by $\mathcal{O}_i$.
Every shard $S_i$ keeps track of local subtransactions that use its corresponding objects in $\mathcal{O}_i$. 

%\paragraph{\bf Shard Network Model}
%\todo{do we need it?}

% {\bf Communication Model :}
% Our communication model is similar to the work in \cite{ adhikari2024spaastable}. We assume that the nodes in a shard are close to each other and connected in a communication network. However, the distance between shards may vary.
% We model the interconnection network between shards as a weighted complete graph (clique graph)
% of shards $G_s$.
% We measure the distance between shards as the number of rounds that are needed 
% until a message is delivered over the network,
% where a round is the time to reach consensus within a shard.
% We consider two communication models as follows:
% \begin{enumerate}
  
% \item
% {\em Uniform communication model}: In this model, any two shards are a unit distance away,
% in the sense that any shard can send or receive information (data/message) within one round. In other words, the shards form a clique where each edge has a weight of $1$.

% \item
% {\em Non-uniform communication model}: In this model, the distance between any two shards ranges from $1$ to $d$, where $d$ is the diameter of the complete graph (clique graph).
% Hence, the edge weights vary from $1$ to $d$. 
% \end{enumerate}

% We assume that shards are connected in a communication graph, and the diameter of the graph is $D$. Moreover, we assume that the distance between the transaction and its accessing objects are in ranges from $1$ to $d\leq D$.

\subsection{\bf Transactions and Subtransactions}
Similar to the works in~\cite{Byshard,adhikari2024spaastable,adhikari2024fast}, we consider transactions $\{ T_1, T_2, \ldots\}$ that are distributed across different shards. Suppose that transaction $T_i$ is generated in a node $v_{T_i}$ within the system, then the {\em home shard} of $T_i$ is the shard containing $v_{T_i}$. 
% The home shard maintains a {\em pending transactions queue} that contains any newly generated transactions, and it's responsible for handling all pending transactions in its queue. 
In this work, we consider transactions that are continuously generated over time. For simplicity and to attain a performance analysis, we assume that each home shard contains at most one transaction at any moment of time, and after the transaction gets processed (either commits or aborts), a new transaction is generated on that home shard.

Similar to work in \cite{Byshard,adhikari2023lockless,adhikari2024spaastable}, we define a transaction $T_i$ as a group of subtransactions $T_{i,a_1},\ldots,T_{i,a_j}$. Each subtransaction $T_{i,a_l}$ accesses objects only in $\mathcal{O}_{a_l}$ and is associated with shard $S_{a_l}$. Therefore, each subtransaction $T_{i,a_l}$ has a respective {\em destination shard} $S_{a_l}$. The home shard sends the transaciton $T_i$ to the leader shard $S_\ell$, which is responsible for processing transaction $T_i$. Then the leader shard of $T_i$ sends subtransaction $T_{i,a_l}$ to shard $S_{a_l}$ for processing, where it is appended to the local blockchain of $S_{a_l}$. The subtransactions within a transaction $T_i$ are independent, meaning they do not conflict and can be processed concurrently. 
% Like previous work in \cite{Byshard,adhikari2024spaastable}, each subtransaction $T_{i,a_l}$ consists of two parts: (i) a condition check, where it verifies whether a condition of the objects in $\mathcal{O}{a_l}$ is satisfied, and (ii) the main action, where it updates the values of the objects in $\mathcal{O}{a_l}$. 

% We provide the example of a transaction and its subtransactions in Appendix~\ref{example-of-transaction-and-subtransaction}.

    Suppose transaction $T_1$ is: 
    {\em ``Transfer 100 coins from account A to account B''}. Let us assume that the accounts of A and B reside on different shards $S_a$ and $S_b$, respectively. $T_1$ splits into the following subtransactions:

\begin{itemize}
    \item[$T_{1,a}$ in $S_a$:] {Condition:} Check if account A has at least 100 coins.\\
    {Action:} Deduct 100 coins from account A.
    \item[$T_{1,b}$ in $S_b$:] {Action:} Add 100 coins to account B.
\end{itemize}

\subsection{Stateless and Stateful Scheduling Models}
\label{sec:stateless-and-statefull-model}
We define two scheduling models to schedule and process the transactions, the stateless and stateful models, which we describe as follows.
% Let's suppose there is a designated {\em leader shard} $S_\ell$ that coordinates the scheduling and processing of transactions. We define two scheduling models as follows:

{\bf Stateless Scheduling Model:} Let's suppose there is a designated {\em leader shard} $S_\ell$ that coordinates the scheduling and processing of transactions.
In this model, the leader shard $S_\ell$ does not maintain the current state of accounts accessed by the transactions~\cite{adhikari2024spaastable,adhikari2024fast, Byshard}. Upon receiving transactions, $S_\ell$ constructs (or extends) a transaction conflict graph and colors the graph using
an incremental greedy vertex coloring algorithm to determine the commit order for each transaction. Then the leader $S_\ell$ splits each transaction into subtransactions based on accessed accounts and sends them to the corresponding {\em destination shards} that hold the relevant account states. Each destination shard maintains the scheduled subtransactions queue $sch_{dq}$ and it picks one color subtransaction from the header of $sch_{dq}$, validates the sub-transactions (e.g., checking account balances) and sends a {\em commit} or {\em abort} vote to the leader. After collecting all votes for a transaction, the leader sends a final decision to each destination shard, which either commits or aborts the subtransactions according to the message received from the leader shard.

For example, suppose $S_\ell$ receives transactions $T_1, T_2, T_3$, each accessing accounts $a, b, c$, located in shards $S_a, S_b, S_c$ respectively (see Figure~\ref{fig:stateless_and_stateful_model} (a)). The leader constructs a conflict graph $G_{\T_\ell}$ and applies a greedy vertex coloring algorithm
%~\cite{busch2022dynamic} 
to define a commit order. It then splits transactions into sub-transactions:
\[
T_1 \rightarrow \{T_{1,a}, T_{1,b}, T_{1,c}\},\quad
T_2 \rightarrow \{T_{2,a}, T_{2,b}, T_{2,c}\},\quad
T_3 \rightarrow \{T_{3,a}, T_{3,b}, T_{3,c}\}
\]
Each destination shard queues the received sub-transactions in a schedule queue $sch_{dq}$ according to the commit order received from $S_\ell$, and it processes one color subtransaction at a time. This means $S_a$ picks $T_{1,a}$, $S_b$ picks $T_{1,b}$ and $S_c$ picks $T_{1,c}$ from head of their queues, check the validity and condition of the subtransaction (such as account balance) and send either commit or abort votes to the leader shard. Then the transaction $T_i$ and its subtransactions ($T_{1,a}, T_{1,b}$ and $T_{1,c}$) are committed or aborted based on the final decision received from the leader shard. Next, each destination shard processes the next color subtransactions, for instance $T_{2,a}$, from $S_a$,  $T_{2,b}$ from $S_b$, and $T_{2,c}$ from $S_c$ (see Figure~\ref{fig:stateless_and_stateful_model} (a)), and this process repeats.
\begin{figure*}
  \centering
  \includegraphics[width=0.99\textwidth]{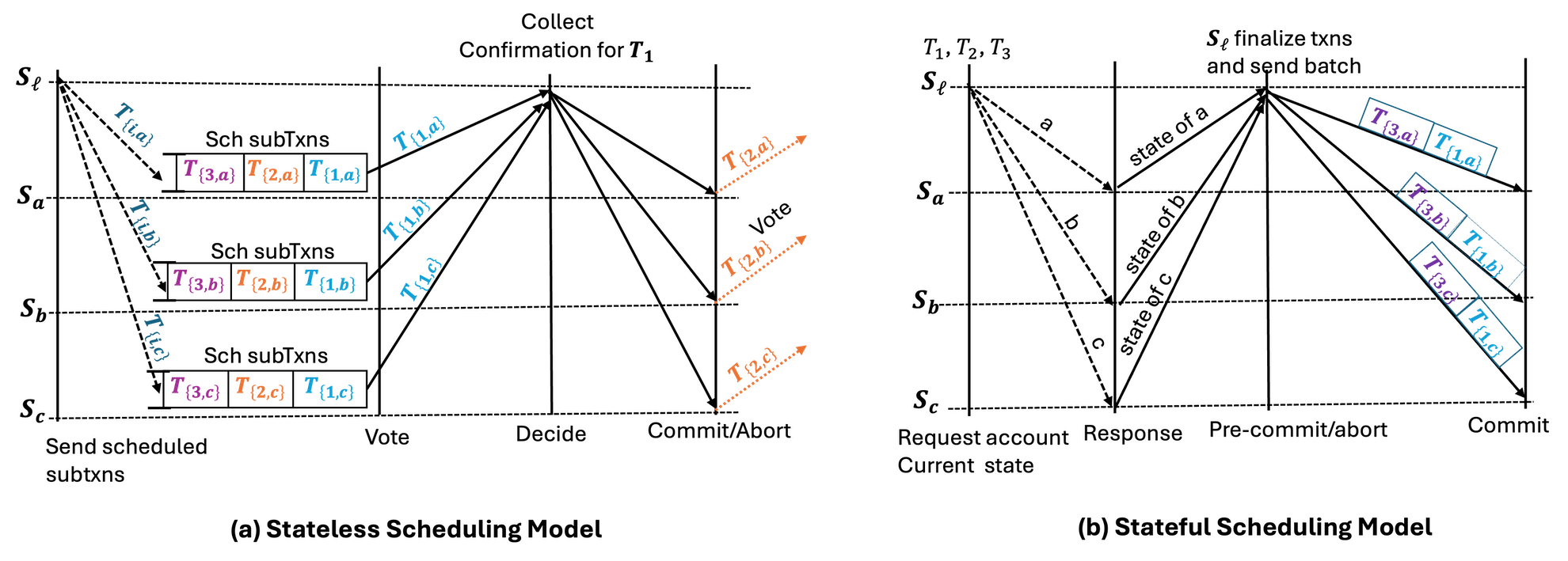}
  \caption{Illustration of Stateless (a) and Stateful (b) Scheduling Models.}
  \label{fig:stateless_and_stateful_model}
  % \vspace{-5mm}
\end{figure*}

{\bf Stateful Scheduling Model:}
In the stateful model, the home shard where a transaction is initially generated sends its transaction information to the leader shard $S_\ell$. Then the leader shard $S_\ell$ stores these transactions (i.e. $\T_1, T_2, T_3$) in its pending transaction queue $PQ_\ell$. Then, the leader shard identifies accounts accessed by transactions and requests their state from corresponding shards $S_a, S_b, S_c$. In other words, before processing the transactions, the leader collects the current state of all accessed accounts from the corresponding destination shards.  Once the account states are gathered, the leader constructs a conflict graph 
%$G_{\T_\ell}$ 
on which it applies the incremental greedy vertex coloring algorithm.
%~\cite{busch2022dynamic} on $G_{\T_\ell}$. 
Then the leader shard performs local \textit{pre-commit} for valid transactions (e.g., $T_1$, $T_3$) and aborts invalid transactions (e.g., $T_2$). After that, $S_\ell$ creates the pre-committed sub-transaction batches:
$
S_a: \{T_{1,a}, T_{3,a}\},\quad
S_b: \{T_{1,b}, T_{3,b}\},\quad
S_c: \{T_{1,c}, T_{3,c}\}
$ for each destination shard $S_a$, $S_b$, $S_c$. Then these pre-committed batches are sent to the respective destination shards. Since the transactions have already been validated, each destination shard can directly commit and append the received pre-committed order to its local blockchain without further interaction with the leader.

The main difference between the stateless and stateful model is that the stateful model requires the leader to be updated about account states which are at remote shards, while the stateless model does not need to be informed about remote accounts.
This additional account information in the stateful model allows for more efficient transaction processing at the expense of added communication complexity.

\subsection{Conflicts and Competitive Ratio}

Two transactions conflict if they attempt to access the same account, and at least one of the two updates the account. 
The subtransactions are processed sequentially at each destination shard. For this reason,
we extend the notion of conflict to all transactions that access account (not necessarily the same) in the same destination shard.

\begin{definition}[Conflict]
 \label{def:conflict}
Transactions $T_i$ and $T_j$ are said to {\em conflict} if they access accounts on the same destination shard $S_k$ and at least one of these transactions 
writes (updates) the account in $S_k$. 
 \end{definition}
 Transactions that conflict should be processed in a sequential manner
to guarantee atomic object update.
In such a case, their respective subtransactions 
should be serialized in the exact same order in every involved shards.
 To resolve the conflict between two transactions $T_i$ and $T_j$ while accessing the same destination shard $S_k$, a scheduler must schedule them one after another in such a way that $T_i$ commits before $T_j$ or vice versa.
 To perform the schedule, we use a {\em conflict graph} such that the nodes are transactions, and an edge represents  a conflict between two transactions.

%  \begin{definition}
%     [{\sc Competitive ratio}] The competitive ratio measures how well a given algorithm performs compared to the best possible offline algorithm. It is the ratio between the total processing time of the proposed algorithm and the minimum processing time that any optimal offline algorithm could achieve~\cite{busch2022dynamic}. The aim is to minimize this ratio, ideally to $1$, but it is typically difficult to achieve. We adopted this definition from~\cite{busch2022dynamic}.
% \end{definition}

% [[[The definitions below are under development]]]

We continue with the definition of competitive ratio for our scheduling models.
The definition below is an adaptation of the competitive ratio used in dynamic execution in software transactional memory \cite{busch2022dynamic}.
Since the future transactions depend on the past execution,
we define the competitive ratio based on any set of transactions
that may appear at any moment of time.
Consider a transaction schedule $S$.
Let $\T_t$ denote the set of all pending transactions (that have not committed or aborted) at time $t$.
Let $t' > t$ denote the time such that all transactions in $\T_t$ finalize (commit or abort).
Let $\tau^*$ denote the optimal time duration to finalize (commit or abort) all the transactions in $\T_t$ if they were the only transactions in the system, processed as a batch.
The approximation ratio for $S$ 
at time $t$ is $r_S(t) = (t' - t) / \tau^*$.
The competitive ratio for $S$ is $r_S = \sup_t r_S(t)$.

\begin{definition}[Algorithm Competitive Ratio]
For online scheduling algorithm $\mathcal{A}$, 
the competitive ratio $r_\mathcal{A}$ is the maximum competitive ratio over all possible schedules $\mathcal{S}$ that it produces, $r_\mathcal{A} = \sup_{S \in \mathcal{S}} r_S$.
(We also say that $\mathcal{A}$ is $r_\mathcal{A}$-competitive.)
\end{definition}

%% file: stateless/single-leader.tex
\section{Stateless Scheduler}
\label{sec:stateless-scheduler}
In this section, we consider the stateless sharding model~\cite{adhikari2023lockless,adhikari2024spaastable,Byshard}, where the leader shard is responsible for coordinating transaction processing and does not gather the current state of account information (see Section~\ref{sec:stateless-and-statefull-model}). We present two transaction scheduling algorithms: the Single-Leader Scheduler and the Multi-Leader Scheduler.

\subsection{Stateless Single-Leader Scheduler}
\label{sec:stateless-single-leader-scheduler}
In this section, we describe and analyze the {\em Stateless Single-Leader  Scheduler}, which operates under a partially synchronous communication model. We assume a designated leader shard $S_\ell$ responsible for determining the transaction schedule. All shards send their transactions to the leader shard, which builds a transaction conflict graph and applies an incremental greedy vertex coloring algorithm %~\cite{busch2022dynamic}
to determine a schedule.

\begin{algorithm*}[t]
% \tiny
% \tiny
\small 
%\smaller[1]
\caption{\sc Stateless Single Leader Scheduler}
\label{alg:stateless-single-leader-scheduler}

txn: transaction; txns: transactions; subTxn: subtransaction; subTxns: subtransactions\;  
$T_i$: txn, $T_{i,j}$: subTxn of $T_i$ for shard $S_j$, $\T_{\ell}$: Set of txns maintained by leader shard $S_\ell$; $is\_busy$: processing flag (initially false at each shard);
Each shard $S_j$ maintains a lexicographically ordered scheduled queue $sch_{dq}$ for subtransactions\;

\BlankLine
\SetKwBlock{WhenTxnGenerated}{\normalfont {\bf Upon generation of a new txn $T_{i}$ at home shard $S_i$}}{}
\WhenTxnGenerated{
    $S_i$ tags local timestamp (ts) to $T_{i}$ and send it to the leader shard $S_\ell$\;\label{alg1:home-shard-send-txn}
}

\BlankLine
\SetKwBlock{WhenLeaderReceivesTxn}{\normalfont {\bf Upon receiving new txn $T_{i}$ at leader shard $S_\ell$}}{}
\WhenLeaderReceivesTxn{
    $S_\ell$ adds $T_{i}$ to txns set $\T_\ell$ and extend transaction conflict graph $G_{\T_\ell}$ with $T_{i}$\;\label{alg1:conflict-graph}
    % If any colored txns $T_x$ with $\text{ts}(T_x) > \text{ts}(T_i)$ exists, cancel its color, prioritize $T_i$ as it is older, and send cancel $T_x$ to respective destinations\;
    If any colored txn $T_x$ exists with $\text{ts}(T_x) > \text{ts}(T_i)$, cancel its color, prioritize $T_i$, and send cancel message for $T_x$ to corresponding destination shards\;\label{alg1:check-timestamp}
    Run incremental greedy coloring on $G_{\T_\ell}$  without altering already scheduled (colored) old txns\;\label{alg1:graph-coloring}
    Split each newly colored $T_{i}$ into subtxns $T_{i,j}$ and
    send to respective destination shard $S_j$\;\label{alg1:split-txn-by-leader}
}

\BlankLine
\SetKwBlock{UponReceiveSubTxn}{\normalfont \textbf{Upon receiving subtransaction} $T_{i,j}$ \textbf{at each destination shard} $S_j$}{end}
\UponReceiveSubTxn{
    Append $T_{i,j}$ in $sch_{dq}$ and order (sort) $sch_{\text{qd}}$ lexicographically according to color\;\label{alg1:sch-queue}
    \If{$is\_busy == \text{false}$}{
        Set $is\_busy = true$\;
        % \tcp{Pick one color subtransaction from the head of the queue; if there are multiple non-conflicting subtransactions of the same color, they can be processed concurrently}
         Let $T_{i,j} \leftarrow$ head of $sch_{dq}$;
    If $T_{i,j}$ is valid and local conditions satisfied, it sends {\em commit vote} to leader shard $S_\ell$;
    Otherwise, it sends {\em abort vote} to $S_\ell$\;\label{alg1:destinaiton-process-txn}
    
    }
}

\BlankLine

\SetKwBlock{WhenLeaderReceivesVotes}{\normalfont {\bf Upon receiving votes for txn $T_i$ at leader shard $S_\ell$}}{}
\WhenLeaderReceivesVotes{
If any abort vote receive for $T_i$ then  it sends {\em confirmed abort} to all corresponding $S_j$ of $T_i$;\label{alg1:leader-coordinate1}
else if all received votes are commit votes, then it
        sends {\em confirmed commit} to corresponding $S_j$\;

    Remove $T_i$ from $\T_{\ell}$ and $G_{\T_{\ell}}$ and send outcome(committed or aborted) to home shard of $T_i$\;\label{alg1:leader-coordinate2}
}

\BlankLine
\SetKwBlock{UponReceiveDecision}{\normalfont \textbf{Upon receiving confirmation for subtxn} $T_{i,j}$ \textbf{at each destination shard} $S_j$}{end}
\UponReceiveDecision{
   
    If the confirmed commit is received, then
        it commit $T_{i,j}$ and append to its local blockchain;
    
    Otherwise, if confirmed abort message received then it abort $T_{i,j}$\;

    If $sch_{dq}$ is not empty, it start to process next subTxn from $sch_{dq}$, else it set $is\_busy = false$\;
   
}

% \BlankLine
% \SetKwBlock{WhenLeaderShardReceivesIgnore}
% {\normalfont \textbf{Upon receiving ignore for $T_{i,j}$ at leader shard $S_\ell$}}{}
% \WhenLeaderShardReceivesIgnore{
%     If no decision sent for $T_i$ then
%         discard vote from $S_j$ for $T_{i,j}$ and
%         send \textit{ignored $T_{i,j}$} to $S_j$\;
%     Else, no action needed; $T_{i,j}$ is already finalized\;
    
% }
\BlankLine
\SetKwBlock{WhenDestinationShardReceivesIcancel}{\normalfont {\bf Upon receiving cancel message for $T_{x,j}$ at destination shard $S_j$}}{}
\WhenDestinationShardReceivesIcancel{
Remove $T_{x,j}$ from $sch_{dq}$; a new color will be received later for $T_{x,j}$ from leader shard $S_\ell$\;
}

\BlankLine
\SetKwBlock{WhenHomeShardReceivesOutcome}{\normalfont {\bf Upon receiving outcome of $T_i$ at home shard $S_i$}}{}
\WhenHomeShardReceivesOutcome{
    Generate next transaction and repeat process\;
}
\end{algorithm*}
The algorithm follows an event-driven approach to schedule and process the transactions. When a new transaction $T_i$ is generated at its home shard $S_i$, then the home shard tags the current timestamp to the transaction $T_i$ and sends the transaction to the leader shard $S_\ell$. Upon receiving $T_i$, the leader adds it to the local transaction set $\T_\ell$ and extends the conflict graph $G_{\T_\ell}$ with this new transaction ($T_{i}$). If $T_i$ is older than any already-colored but uncommitted transactions (say $T_x$), the leader cancels the color of those newer transactions, notifies the relevant shards, and reprocesses them later. This ensures older transactions are prioritized, avoiding starvation. The leader then runs an incremental greedy vertex coloring algorithm~\cite{busch2022dynamic} to assign colors to all newly received transactions, without modifying the colors of already scheduled old transactions. This ensures that the processing time of already scheduled transactions is not affected by newly generated transactions. Note that a newer transaction might receive a lower color than an older one because the new one does not conflict with any other transaction (except one old transaction), while the old transaction conflicts with others as well. To prevent this and ensure a fair execution order, we assign each new transaction a color no lower than the smallest color among pending old transactions. This approach guarantees progress because at each time step, the lowest possible color will increase over time. After coloring and determining the schedule, each transaction is then split into subtransactions $T_{i,j}$ based on the destination shards it accesses, and these subtransactions are sent to the corresponding shards $S_j$ for processing.

% Note that it is possible that a newer gets a lower color, because the new does not conflict with any other transaction (except the old), while the old may conflict with another old. The solution would be to set the lowest possible color that a transaction may get to the lowest color of the pending old transactions. This guarantees progress, because at each time step the  lowest possible color will increase

Each destination shard $S_j$ maintains a local scheduled queue $sch_{dq}$ (consisting of subtransactions that have been scheduled but not yet committed) and appends incoming subtransactions into $sch_{dq}$, which stores subtransactions in the order of their assigned color. To handle partial synchrony, each destination shard $S_j$ uses a busy flag to track whether it is currently processing (in-transit and not committed yet) a subtransaction. If the shard is not busy, it picks one subtransaction from the head of the queue and validates it (e.g., checking conditions like account balance). If the subtransaction is valid, the shard $S_j$ sends a \emph{commit vote} to the leader $S_\ell$; otherwise, it sends an \emph{abort vote}. Once the leader shard receives votes from all relevant destination shards for a transaction $T_i$, it decides whether the transaction should be committed or aborted. If all subtransactions vote to commit, the leader sends a \emph{confirmed commit} to each destination shard; otherwise, if any one of the shard send an abort vote, it sends a \emph{confirmed abort}. After the decision, the transaction $T_i$ is removed from the conflict graph $G_{\T_\ell}$ and the transaction set $\T_\ell$, and the outcome (committed or aborted) is sent to the home shard of $T_i$.

Upon receiving the confirmed decision, each destination shard either commits the subtransaction by appending it to the local blockchain or aborts it. If the scheduled queue is not empty, the shard continues processing the next subtransaction. If the queue becomes empty, the shard marks itself as not busy. Finally, upon receiving the outcome from the leader, the home shard generates a new transaction and repeats the process. This single leader scheduling approach ensures conflict-free execution while preserving consistency and fairness in transaction processing across shards.

% The correctness (Safety and Liveness) analysis of Algorithm~\ref{alg:stateless-single-leader-scheduler} is deferred to Appendix~\ref{app:safety-and-liveness-analysis-of-stateless-single-leader-scheduler}.

\subsubsection{Correctness Analysis of Stateless Single-Leader Scheduler (Algorithm~\ref{alg:stateless-single-leader-scheduler})} 
\label{app:safety-and-liveness-analysis-of-stateless-single-leader-scheduler}

% \subsubsection{Correctness Analysis of Single Leader Scheduler (Algorithm~\ref{alg:stateless-single-leader-scheduler})} 

Our proposed scheduling algorithm works on a partial-synchronous communication model; for the sake of analysis only, we consider the synchronous communication mode.

% The proof of Lemma~\ref{lemma:stateless-single-leader-safety} is deferred to Appendix~\ref{app:proof-of-lemma-stateless-single-leader-safety}.

% The proof of Lemma~\ref{lemma:stateless-single-leader-liveness} is deferred to Appendix~\ref{app:proof-of-lemma:stateless-single-leader-liveness}

\begin{lemma}[Safety]
\label{lemma:stateless-single-leader-safety}
   If two transactions conflict with each other in Algorithm~\ref{alg:stateless-single-leader-scheduler}, then they will commit in different time slots, and the local chain produced by Algorithm~\ref{alg:stateless-single-leader-scheduler} ensures blockchain serialization.
\end{lemma}
\begin{proof}
We prove this by induction (analyzing) the execution of Algorithm~\ref{alg:stateless-single-leader-scheduler}, where each home shard sends its transaction to the leader shard (Line~\ref{alg1:home-shard-send-txn}), and the leader shard constructs the transaction conflict graph $G_{\T_\ell}$ (Line~\ref{alg1:conflict-graph}). Then the leader used the incremental greedy vertex coloring algorithm~\cite{busch2022dynamic} on the conflict graph $G_{\T_\ell}$ (Line~\ref{alg1:graph-coloring}). As conflicting transactions share an edge in $G_{\T_\ell}$, they are assigned different colors and are processed in different time slots, which provides the valid commit order. Moreover, each color corresponds to a unique serialization time slot. 
% Line~\ref{alg1:send-txn}
% Line~\ref{alg1:check-timestamp}
The leader shard splits the transaction into subtransactions and sends them to the destination shard after coloring (see Line~\ref{alg1:split-txn-by-leader}), then each destination shard keeps that ordering in the schedule queue ($sch_{dq}$) and process subtransactions one by one according to the color they get (see Line~\ref{alg1:sch-queue}-\ref{alg1:destinaiton-process-txn}), which guarantees the consistent schedule order in each shard. Moreover, the leader shard coordinates to commit the subtransactions in each destination shard, which ensures the consistent commitment (see Line~\ref{alg1:leader-coordinate1}-\ref{alg1:leader-coordinate2}). As the subtransactions are committed according to the color they receive, and each color corresponds to a globally consistent time slot, this provides global serialization.
\end{proof}

\begin{lemma}[Liveness]
\label{lemma:stateless-single-leader-liveness}
    Algorithm~\ref{alg:stateless-single-leader-scheduler} guarantees that every generated transaction will eventually be either committed or aborted.
\end{lemma}
\begin{proof}
We prove liveness by induction, showing that every transaction $T_i$ is either committed or aborted in finite time. Each new transaction $T_i$ is sent to a leader shard $S_\ell$ (Line~\ref{alg1:home-shard-send-txn}), which adds it to the set $\T_\ell$ and the conflict graph $G_{\T_\ell}$. If $T_i$ is older than any already colored but not committed transaction $T_x$, the algorithm cancels the color of $T_x$ and re-colors the graph (Line~\ref{alg1:check-timestamp}). Coloring is performed incrementally (Line~\ref{alg1:graph-coloring}) and preserves the colors of previously scheduled transactions. Thus, older transactions are always prioritized, and no transaction is indefinitely prevented from being scheduled due to newer ones. Note that a newer transaction might receive a lower color than an older one because the new one does not conflict with any other transaction (except one old transaction), while the old transaction conflicts with others as well. To prevent this and ensure a fair execution order, we assign each new transaction a color no lower than the smallest color among pending old transactions. This approach guarantees progress because at each time step, the lowest possible color will increase over time.

Moreover, once $T_i$ is colored, its subtransactions are sent to the respective destination shards (Line~\ref{alg1:split-txn-by-leader}), where they are placed into a queue $sch_{dq}$ sorted by color (Line~\ref{alg1:sch-queue}). Each shard processes one color group at a time, controlled by a busy flag. After finishing one subtransaction (commit or abort), the shard proceeds to the next one in the queue. Since every color is eventually dequeued, and subtransactions are processed in order, every scheduled subtransaction is eventually processed.
Thus, every transaction is either committed or aborted in a finite time, and this proves the liveness.
\end{proof}

\begin{corollary}
    From Lemma~\ref{lemma:stateless-single-leader-safety} and Lemma~\ref{lemma:stateless-single-leader-liveness}, Algorithm~\ref{alg:stateless-single-leader-scheduler} ensures the safety and liveness of the transactions.
\end{corollary}

\subsubsection{Performance Analysis of Single-Leader Scheduler (Algorithm~\ref{alg:stateless-single-leader-scheduler})} 
Our proposed scheduling algorithm works on a partial-synchronous communication model; for the sake of performance analysis only, we consider the synchronous communication mode.
In the following, we analyze the time units required to process transactions by Algorithm \ref{alg:stateless-single-leader-scheduler}. We are focusing on the time period after the leader shard has determined the schedule for the transactions.
In the synchronous case, a time unit is the time to send a message along an edge of unit weight.
In the single-leader case, $d$ is sensitive to the position of the leader and $d$ denotes the maximum distance between any of the involved shards (home, destination shards, leader shard).
In the multi-leader case, the distance to the leaders is not included in the definition of $d$.

\begin{theorem}
\label{theorem:stateless-general-graph-competative-ratio}
    [General Graph] In the General graph, where the transactions, their accessing objects, and the leader are at most $d$ distance away from each other, Algorithm~\ref{alg:stateless-single-leader-scheduler} has $O(d\cdot \min\{k, \sqrt{s}\})$ competitive ratio.
\end{theorem}

\begin{proof}
    Consider a set of transactions $\T$ generated at or before time $t$ that are still pending (neither committed nor aborted) at time $t$.
%Let $\T' \subseteq \T$ denote the transactions that will be colored at the leader,
%while the remaining $\T'' = \T \setminus \T'$
%have already received a color.
%The choice of color for each $\T'$
Let $G_{\T}$ denote the conflict graph for $\T$,
where two transactions conflict if they have a common destination shard.
Since we use greedy coloring to color $G_{\T}$,
the number of distinct colors assigned to the transactions in $\T$ depends only on the coloring of $G_{\T}$, and not on the colors of the transactions that have been finalized (committed or aborted) before $t$. (This holds since transactions in $\T$ may use smaller 
 colors of transactions committed before $t$.)

Let $l_i$ denote the number of transactions in $\T$ that use objects in shard $S_i$.
Let $l = \max l_i$. 
We have that $l$ is a lower bound on the time 
that it takes to finalize (commit or abort) the transactions in $\T$, since at least $l$ subtransactions need to serialize in a destination shard. 

First, consider the case where $k \leq \sqrt s$.
We have that each transaction conflicts with at most $k l$
other transactions. Hence $G_{\T}$ can be colored with
at most $k l + 1$ colors.
The distance between a transaction (home shard) and its accessing objects(destination shards) is at most $d$ away, and to commit subtransactions after being scheduled, Algorithm \ref{alg:stateless-single-leader-scheduler} takes $3$ steps of interactions (for each color) between the leader shard and the destination shard. This means each color corresponds to the $3d$ time units.
Thus, it takes at most $(kl+1)3d = O(kld)$ time units to confirm and commit the transactions.
Hence,
for transactions $\T$,
the approximation of their finalization time 
is $O(kld / l) = O(kd)$.

Next, consider the case $k > \sqrt s$.
We can write $\T' = A \cup B$,
where $A$ are the transactions which access at most $\sqrt s$
destination shards, while $B$ are the transactions which access more than $\sqrt s$ destination shards.
Each transaction in $A$ conflicts with at most $l \sqrt s$ other transactions. Hence, the transactions in $A$ need at most $l \sqrt s + 1$ distinct colors.
The transactions in $B$ can be serialized, requiring at most $|B|$ distinct colors.
Hence, the conflict graph $G_T$ can be colored with at most
$l \sqrt s + 1 + |B|$ colors, which implies a schedule of 
length $O(d (l \sqrt s +  |B|))$ steps to finalize the transactions $\T$.
Since each transaction in $B$ accesses more than $\sqrt s$ shards,
there is a shard accessed by more than $(|B| \sqrt s) / s = |B| / \sqrt s$ transactions. 
Thus, $l > |B| / \sqrt s$.
Hence,
for transactions $\T$,
the approximation of their finalization time is 
$O(d(l \sqrt s +  |B|)/ l) = O(d \sqrt s + d |B| / l) = O(d \sqrt s +  d \sqrt s) = O(d \sqrt s)$.

Therefore, combining the approximations for the cases $k \leq \sqrt s$ and $k > \sqrt s$,
we have that the combined approximation for the finalization time for $\T$ is $O(d \cdot \min\{k,\sqrt s\})$.
Since $t$ is chosen arbitrarily,
we have that the competitive ratio of Algorithm~\ref{alg:stateless-single-leader-scheduler} is $O(d \cdot \min\{k,\sqrt s\})$.
\end{proof}

Suppose that shards are connected in a clique graph with unit distance, where every shard is connected to every other shard with unit distance. So in this case $d=1$. Then from Theorem~\ref{theorem:stateless-general-graph-competative-ratio}, Algorithm~\ref{alg:stateless-single-leader-scheduler} has an $O(\min\{k, \sqrt{s}\})$ competitive ratio for a clique graph with unit distance. Thus, we have:
\begin{corollary}[Unit Distance Clique]
\label{cor:unit_clique}
 Algorithm \ref{alg:stateless-single-leader-scheduler} has an $O(\min\{k, \sqrt{s}\})$ competitive ratio for a clique graph with unit distance.
\end{corollary}

We continue to show that it is an NP-hard problem to approximate the optimal transaction schedule.
Thus, the provided bound in Corollary \ref{cor:unit_clique},
is the best we can do with a polynomial time scheduling algorithm.
The result below applies to both the stateful and stateless model.

\begin{theorem}
For all $\epsilon > 0$, it is an NP-hard problem to produce a transaction schedule that achieves a competitive ratio $(\min\{k, \sqrt{s}\})^{1 - \epsilon}$.
\end{theorem}

\begin{proof}
We will use a reduction from vertex coloring.
For all $\epsilon > 0$, the problem of approximating the chromatic number of a graph with $n$ nodes within a factor $n^{1-\epsilon}$ is NP-hard~\cite{10.1145/1132516.1132612}.

Consider an instance of vertex coloring on a graph $H = (V_H, E_H)$ with $n$ nodes.
We can transform the vertex coloring instance $H$ to a scheduling problem instance on a graph shard $G_s$ with $s = |E_H|$ shards,
such that
$G_s$ is a synchronous clique with unit distances between the shards.
Furthermore, each edge of $E_H$ corresponds to a unique node of $G_s$.

Let $\T$ be a set of $n$ transactions,
all  generated concurrently at time $t = 0$,
such that each node $v_i \in V_H$ is mapped to 
transaction $T_i \in \T$.
For each edge $(v_i,v_j) \in E_H$ we create a conflict between respective transactions $T_i$ and $T_j$ by making the transactions access a common object in the unique shard of $G_s$ that corresponds to edge $(v_i,v_j)$. 
Let $G_{\T}$ be the respective conflict graph for the transactions $\T$.
The conflict graph $G_{\T}$ is isomorphic to $H$.

A correct execution schedule for $\T$ (which gives a valid serialization of the transactions in $\T$)
can be represented as a DAG where nodes are transactions and transaction $T_i$ points to $T_j$ if they conflict and $T_i$ executes first in the respective common destination shard with $T_j$.
Then, a layering of the DAG nodes starting from source nodes provides a unique time step for each transaction, so that conflicting transactions receive different time steps.
Thus, an execution schedule of the transactions in $\T$
gives a valid vertex coloring of the nodes in $G_{\T}$ which provides a valid coloring for $H$.
The best length of the transaction schedule given from the DAG, is equal to the number of colors that can be assigned to $H$.

Since $|E_H| \leq n(n-1)/2$, we have that $s = O(n^2)$.
Each transaction conflicts with at most $k \leq n-1$ other transactions.
Therefore, given $k$ and $s$, we can create the reduction from graph coloring for $n = \min(k,\sqrt s)$.
Consequently, the NP-hardness of the scheduling problem in $G_s$ follows from the NP-hardness of the reduced graph coloring problem with $n = \min(k, \sqrt s)$. 
\end{proof}

%% file: stateless/multi-leader.tex
\subsection{Multi-Leader Scheduler}
\label{sec:stateless-multi-leader-scheduler}
This section provides the multi-leader scheduler where multiple leaders schedule and process the transactions, distribute the congestion, and load across different leaders.
The multi-leader approach allows adaptation to the value $d$ without requiring knowledge of $d$.
Also, here the value $d$ depends only on the maximum distance between the home and destination shards (without involving distances to the leaders). Therefore, 
the value of $d$ captures better the locality of the transactions,
and the resulting schedule allows for shorter messages between home and destination shards.
The concepts that we introduce for this algorithm will play a key role for the development of the stateless multi-leader algorithm.

\subsubsection{Shard Clustering}
\label{sec:shard-clustering}
In the multi-leader scheduler, shards are distributed across the network, and the distance between the home shard of the transaction and its accessing objects (destination shards) ranges from $1$ to $D$, where $D$ is the diameter of the shard graph. Let us suppose shards graph $G_s$ constructed with $s$ shards,
where the weights of edges between shards denote the distances between them.
We consider that $G_s$ is known to all the shards.
% Consider $D$ as the diameter of $G_s$ (i.e., $D$ is the maximum distance between any two shards).
We define $z$-neighborhood of shard $S_i$ as the set of shards within a distance of at most $z$ from $S_i$. Moreover, the 0-neighborhood of shard $S_i$ is the $S_i$ itself.

%Same as previous work by Busch \Tl \cite{busch2022dynamic}, 
We consider that our multi-leader scheduling algorithm uses a hierarchical decomposition of $G_s$
which is known to all the shards and calculated before the algorithm starts.
This shard clustering (graph decomposition) is based on the clustering techniques in \cite{gupta2006oblivious} and which were later used in \cite{sharma2014distributed,busch2022dynamic,adhikari2024spaastable}.
We divide the shard graph $G_s$ into the hierarchy of clusters with $H_1 = \lceil \log D \rceil +1$ layers (logarithms are in base 2), and a layer is a set of clusters, and a cluster is a set of shards. 
%We use a notion of the {\em weak diameter} of the cluster where the distance between shards within a cluster is measured with respect to $G_s$, not within the subgraph induced by the cluster.
%There are cluster construction techniques \cite{gupta2006oblivious,sharma2014distributed} which give a hierarchy of clusters with $H_1$ layers and the diameter of each cluster 
Layer $q$, where $0\leq q < H_1$, is a sparse cover of $G_s$
such that:
\begin{itemize}
    \item Every cluster of layer $q$ has (strong) diameter of at most $O(2^q\log s)$.
    \item Every shard participates in no more than $O(\log s)$ different clusters at layer $q$.
    \item For each shard $S_i$ there exists a cluster at layer $q$ which contains the $(2^q-1)$-neighborhood of $S_i$ within that cluster.
\end{itemize}
 
%The clusters have a {\em strong diameter}, 
%which measures distances within a cluster on the induced subgraph of $G_s$.
%, but as per \cite{busch2022dynamic}, these still fulfill our requirements by providing weak diameter properties. 
%A hierarchical sparse cover construction is available in literature \cite{sharma2014distributed}, 
For each layer $q$,
the sparse cover construction in \cite{gupta2006oblivious}
is actually obtained as a collection of $H_2 = O(\log s)$ partitions of $G_s$. 
%which ensures that in each layer, $l$, a shard $S_i$ is a member of at most $H_2 = O(\log s)$ distinct clusters.
These $H_2$ partitions are ordered as sub-layers of layer $q$ labeled from $0$ to $H_2-1$.
%, where each sub-layer represents a cluster at height $l$ of $G_s$. 
A shard might participate in all $H_2$ sub-layers but potentially belongs to a different cluster at each sub-layer.
At least one of these $H_2$ clusters at layer $q$ contains the whole $2^q-1$ neighborhood of $S_i$.

In each cluster at layer $q$, a leader shard $S_\ell$ is specifically designated such that the leader’s $(2^q-1)$-neighborhood is in that cluster.
As we give an idea of layers and sub-layers, we define the concept of height as a tuple $h= (h_1, h_2)$, where $h_1$ denotes the layer and $h_2$ denotes the sub-layer. Similar to \cite{sharma2014distributed,busch2022dynamic,adhikari2024spaastable}, heights follow lexicographic order.

The {\em home cluster} for each transaction $T_i$ is defined as follows: suppose $S_i$ is the home shard of $T_i$,
and $z$ is the maximum distance from $S_i$ 
to the destination shards that will be accessed by $T_i$;
%from a position of $T$ (i.e., from the home shard of $T$). 
%The home cluster for $T$ 
the home cluster of $T_i$ is the lowest-layer (and lowest sub-layer) cluster 
in the hierarchy that contains $z$-neighborhood of $S_i$.
Each home cluster consists of one dedicated leader shard, which will handle all the transactions that have their home shard in that cluster (i.e., transaction information will be sent from the home shard to the cluster leader shard to determine the schedule).

\begin{figure*}[t]
    \centering
    \includegraphics[width=0.99\textwidth]{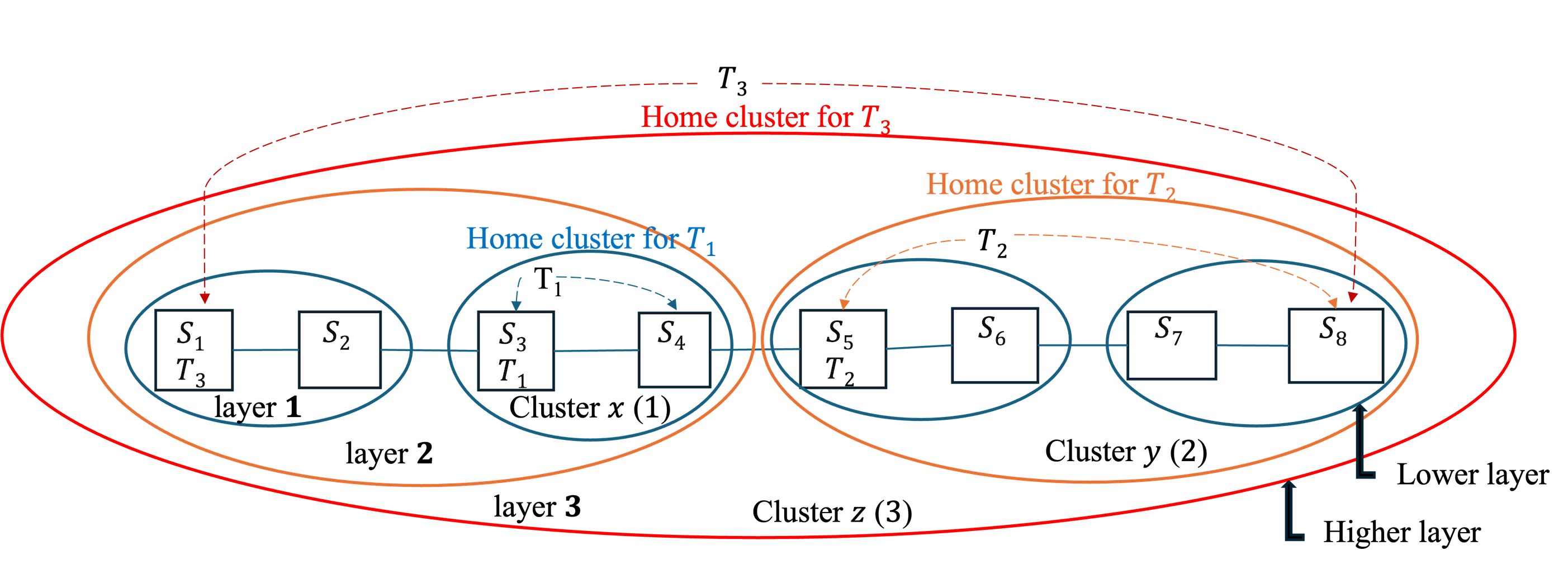}

\caption{Simple example of cluster decomposition of shard graph $G_s$. }
\label{fig:hierarchical-clustering}
% \cb{remove bottom "(a) (b)" increase font in boxes}
\end{figure*}
Figure~\ref{fig:hierarchical-clustering} shows an example of hierarchical clustering, 
assuming shards are connected as if they are in a line, 
where edges in the line have low weights and edges not in the line have large weights.
(We omit the sublayers to  simplify the example.) 
Transaction $T_1$ resides in shard $S_3$ and has home cluster {\em x} at layer $1$. The reason for the home cluster $x$ selection is that $T_1$ accesses an object in  $S_3$ and $S_4$, and both of them are in cluster $x$, and $x$ is the lowest layer cluster including $S_3$ and $S_4$.
Similarly, suppose transaction $T_2$, which resides in $S_5$, has home cluster $y$ at layer $2$, because $T_2$ accesses an object in $S_5$ and $S_8$, and $y$ is the lowest layer cluster that includes both $S_5$ and $S_8$. Similarly, $T_3$ has home cluster $z$ at layer $3$.
% , and each cluster has one designated leader shard that coordinates to process all the transactions belonging to that cluster.

\subsubsection{Stateless Multi-Leader Scheduler}

We consider a hierarchical clustering of the shard graph $G_s$, which is assumed to be globally known by all shards. Each cluster $C$ in this hierarchy is characterized by a unique height $(q, r)$ which corresponds to its layer $q$ and sublayer $r$, and each cluster $C$ has its designated leader shard $S_\ell$. The leader shard is responsible for scheduling and coordinating the processing of all transactions whose home cluster is $C$. Each home shard $S_i$ maintains a local timestamp $ts$ to tag newly generated transactions. Additionally, each destination shard $S_j$ maintains a local scheduling queue $sch_{dq}$ and lexicographically orders for the incoming subtransactions using the tuple $(ts, q, r, color)$, where $color$ is an integer assigned to the transaction by the leader shard $S_\ell$ through vertex coloring. Algorithm~\ref{alg:stateless-multi-leader-scheduler} invokes Algorithm~\ref{alg:stateless-single-leader-scheduler} in each cluster $C$ to process their transactions.

Algorithm~\ref{alg:stateless-multi-leader-scheduler} works in a partially synchronous model and follows an event-driven execution by message passing. When a new transaction $T_i$ is generated at its home shard $S_i$, then the home shard $S_i$ determines the lowest cluster $C$ at height $(q, r)$ that includes both $S_i$ and all of the destination shards accessed by $ T_i$. Moreover, the transaction is tagged with its local timestamp $ts$, along with the cluster identifiers $q$ and $r$, and is then sent to the cluster’s leader shard $S_\ell$. 

Upon receiving new transaction(s) $T_{i}$, the leader shard $S_\ell$ of cluster $C$ invokes Algorithm~\ref{alg:stateless-single-leader-scheduler} to process their transactions, where leader shard $S_\ell$ adds  $T_{i}$ to the transaction set $\T_C$ of cluster $C$ and updates the corresponding transaction conflict graph $G_{\T_C}$ to incorporate the new transaction  $T_{i}$. Then the leader shard used an incremental greedy vertex coloring algorithm~\cite{busch2022dynamic} to assign a color only to the newly received transaction without affecting already colored (scheduled) transactions. Once colored, the transaction is split into subtransactions $T_{i,j}$, and sent to the respective destination shard $S_j$.

Since multiple leader shards process their transactions concurrently by invoking the Algorithm~\ref{alg:stateless-single-leader-scheduler}, destination shards may receive the subtransactions from different clusters simultaneously. To handle this, we modify the parameters and processing technique of Algorithm~\ref{alg:stateless-single-leader-scheduler} as follows:  each destination shard $S_j$ maintains a scheduled subtransactions queue $sch_{qd}$, which is ordered lexicographically by the tuple $(ts,q,r,color)$. The additional parameters $(ts,q,r)$ denote the timestamp $ts$, and hierarchical cluster heights (layers $q$ and sublayers $r$) in the shard graph $G_s$.  Moreover, each destination shard $S_j$ processes its subtransactions from the head of $sch_{dq}$ following the steps in Algorithm~\ref{alg:stateless-single-leader-scheduler} with the modified ordering criteria.

\begin{algorithm*}[t]
%\smaller[1]
\small 
\caption{\sc Stateless Multi-Leader Scheduler}
\label{alg:stateless-multi-leader-scheduler}
Assume all shards know a hierarchical cluster decomposition of $G_s$\;
Each cluster $C$ is associated with a unique height $(q, r)$ and has a designated leader shard $S_\ell$\;
Each shard $S_j$ maintains a lexicographically ordered queue $sch_{dq}$ for subtransactions\;

\BlankLine
\SetKwBlock{OnTxnGenerated}
{\normalfont {\bf Upon generation of a new transaction $T_{i}$ at home shard $S_i$}}{}
\OnTxnGenerated{
    $S_i$ tags a local timestamp (ts) to $T_i$ and identifies the destination shards accessed by $T_i$\;
     $S_i$ selects the lowest cluster $C$ with height $(q, r)$ that contains $T_i$ and all its destination shards\;
    $S_i$ sends $T_i$ to the leader shard $S_\ell$ of cluster $C$\;
}

\BlankLine
\SetKwBlock{OnTxnReceivedAtLeader}
{\normalfont {\bf Upon receiving transaction $T_i$ at the leader shard $S_\ell$ of cluster $C$}}{}
\OnTxnReceivedAtLeader{
    The leader shard $S_\ell$ of each cluster $C$ invokes Algorithm~\ref{alg:stateless-single-leader-scheduler} to schedule and process their transactions. This means each cluster $C$ invokes Algorithm~\ref{alg:stateless-single-leader-scheduler} to process their transactions\;

   \tcp{Since multiple clusters process their transactions concurrently, each with its own leader, destination shards may receive subtransactions from different clusters simultaneously.} 
   \BlankLine
   \SetKwBlock{MultiLeaderHandling}
{\normalfont {\bf To handle subtransactions from multiple clusters (leaders):}}{}

\MultiLeaderHandling{
    Each destination shard $S_j$ maintains a scheduled subtransactions queue $sch_{dq}$ ordered lexicographically by the tuple $(ts, q, r, color)$.
    The additional parameters $(ts,q, r)$ reflect the hierarchical cluster heights (layers and sublayers) in the shard graph $G_s$\;
    
    Each destination shard $S_j$ processes their subtransactions from the head of $sch_{dq}$ following the rules in Algorithm~\ref{alg:stateless-single-leader-scheduler}, with the modified ordering criteria\;
}
}
\end{algorithm*}

% Each destination shard $S_j$ maintains a queue $sch_{qd}$ for pending subtransactions(which are scheduled but not committed yet) and destination shard marks itself as \textit{busy} if it is currently processing one of the subtransactions. When it receives the subtransactions from the leader, it appends those subtransactions to the schedule queue $sch_{qd}$, and then the queue is ordered lexicographically by $(lt, q, r, color)$. If the destination shard is not busy, it dequeues(picks) the highest-priority subtransaction, validates it (e.g., checking conditions such as account balance), and sends a \emph{commit vote} if subtransaction is valid and condition is satisfied or an \emph{abort vote} (if the subtransaciton is invalid or condition is not satisfy) to the leader shard $S_\ell$.
Additionally, if the destination shard is busy and receives a new subtransaction $T_{i',j}$ such that $ts(T_{i',j}) < ts(T_{i,j})$ in lexicographic order, this means $T_{i',j}$ has a higher priority where $T_{i,j}$ is the currently processed (but not committed) subtransaction, then the shard give priority to $T_{i',j}$ by sending an \emph{ignore} $T_{i,j}$ message to its leader, indicating that a higher-priority transaction (subtransaction $T_{i',j}$) should proceed first. Then, when the leader receives an \emph{ignore} $T_{i,j}$ message for a subtransaction $T_{i,j}$ and the decision for $T_i$ has not yet been made (i.e., not all votes have been received), the leader discards the vote from $S_j$ and replies with an \emph{ignored} $T_{i,j}$ message to the destination shard $S_j$. If the decision has already been made (i.e, confirm commit or confirm abort) by the leader shard, then no further action is taken for particular subtransaction $T_{i,j}$ at the leader shard $S_\ell$. Then, if the destination shard $S_j$ receives an \emph{ignored} message for $T_{i,j}$, then it reinserts $T_{i,j}$ into the scheduled queue, reorders the queue lexicographically, and resumes processing from the head.

Finally, when the home shard $S_i$ receives the final outcome of its transaction $T_i$, it generates a new transaction and sends it to the corresponding cluster leader shard, and the process repeats. This multi-leader scheduling framework ensures conflict-free and consistent execution by leveraging lexicographic ordering over the tuple $(ts, q, r, color)$, and maintains the fairness and parallelism across shards in the presence of partial synchrony.

% The correctness analysis of Multi-Leader Scheduler (Algorithm~\ref{alg:stateless-multi-leader-scheduler}) is deferred to Appendix~\ref{app:safety-and-liveness-of-stateless-multi-leader-scheduler}.

\subsubsection{Correctness Analysis of Stateless Multi-Leader Scheduler (Algorithm~\ref{alg:stateless-multi-leader-scheduler})}
\label{app:safety-and-liveness-of-stateless-multi-leader-scheduler}

\begin{lemma}[Safety]
\label{lemma:stateless-multi-leader-safety}
   If two transactions conflict with each other in Algorithm~\ref{alg:stateless-multi-leader-scheduler}, then they will commit in different time slots, and the local chain produced by Algorithm~\ref{alg:stateless-multi-leader-scheduler} ensures blockchain serialization.
\end{lemma}
\begin{proof}
This proof follows the similar reasoning discussed in Lemma~\ref{lemma:stateless-single-leader-safety}, where the leader of each cluster used an incremental greedy vertex coloring algorithm~\cite{busch2022dynamic} to color the transaction conflict graph $G_{\T_C}$ so that conflicting transactions get different colors. 
Moreover, each destination shard $S_j$ maintains a queue $sch_{dq}$ of pending subtransactions lexicographically ordered by the tuple $(ts, q, r, color)$, which is consistent across all destination shards.
Each shard processes subtransactions from the head of the queue, ensuring a consistent order of execution that respects the coloring-based serialization. Thus, conflicting transactions are guaranteed to be processed in separate time slots, and all shards maintain the same lexicographic ordering commit order, which ensures global blockchain serialization.
\end{proof}

\begin{lemma}[Liveness]
\label{lemma:stateless-multi-leader-liveness}
    Algorithm~\ref{alg:stateless-multi-leader-scheduler} guarantees that every generated transaction will eventually be committed or aborted.
\end{lemma}
\begin{proof}
This follows the similar reasoning of the proof of Lemma~\ref{lemma:stateless-single-leader-liveness}, where each cluster $C$ constructs and maintains a conflict graph $G_{\T_C}$ and incrementally colors the vertices using a greedy vertex coloring algorithm~\cite{busch2022dynamic}. The coloring is incremental and does not modify the color assignments of previously scheduled old transactions, which prevents starvation of older transactions.
Each destination shard processes the subtransactions in the lexicographic order of $sch_{dq}$ based on the tuple $(ts, q, r, \text{color})$. A `busy` flag at each shard ensures that only one color (i.e., scheduling round) is active at any time. Once all subtransactions of the current color are processed (i.e., either committed or aborted), the shard proceeds to the next color. Moreover, if a subtransaction with an earlier lexicographic order arrives while a later one is being processed (a possibility in partially synchronous settings), it is reinserted into the queue and priority is given to the older transaction appropriately. 
Therefore, the algorithm guarantees that all scheduled transactions eventually reach a decision. No transaction is indefinitely blocked, ensuring that each transaction is eventually either committed or aborted. Hence, Algorithm~\ref{alg:stateless-multi-leader-scheduler} satisfies liveness.
\end{proof}

\begin{corollary}
    From Lemma~\ref{lemma:stateless-multi-leader-safety} and Lemma~\ref{lemma:stateless-multi-leader-liveness}, Algorithm~\ref{alg:stateless-multi-leader-scheduler} ensures the safety and liveness of the transactions.
\end{corollary}

\subsubsection{Performance Analysis of Multi-Leader Scheduler (Algorithm~\ref{alg:stateless-multi-leader-scheduler})}
The multi-leader scheduler is the extended version of the single-leader scheduler (Algorithm~\ref{alg:stateful-single-leader-scheduler}) while introducing an additional overhead cost due to its shard (hierarchical) clustering structure and comes from the layers and sublayers of the clusters.

% % This overhead is 
% $O(\log D \log s)$, where 
% $O(\log D)$ accounts for the layered hierarchy and 
% $O(\log s)$ comes from the sublayers. Here, 
% $D$ represents the diameter of shard graph $G_s$, and 
% $s$ denotes the total number of shards. 

% Moreover, each cluster $C_i$ at height $(i,j)$ has its own diameter $D_i \leq D$, where $\mathfrak{D}$ is the maximum diameter among all clusters, which represents the upper bound on local processing time and communication delay between any two shards. This upper bound $\mathfrak{D}$ is defined in Section~\ref{preliminaries}.
% Despite this added cost, the multi-leader approach improves scalability and optimizes distributed transaction scheduling.
\begin{theorem}
\label{theorem:stateless-upper-bound-multi-leader-scheduler}
In Multi-leader scheduler (Algorithm \ref{alg:stateless-multi-leader-scheduler}), 
 where the transactions and their accessing objects are at most $d$ distance away from each other, Algorithm~\ref{alg:stateless-multi-leader-scheduler} has $O(d\log^2 s  \cdot\min\{k, \sqrt{s}\})$ competitive ratio.
\end{theorem}
\begin{proof}
%If we only had one cluster $C$ where the distance between transaction and their accessing shard is at most $d$ distance away, then this case is similar to the stateless single leader scheduler (Algorithm~\ref{alg:stateless-single-leader-scheduler}), and we get the same competitive ratio as presented in Theorem~\ref{theorem:stateless-general-graph-competative-ratio}.
In the multi-layer scheduler, we need to consider the transactions from all layers and sublayers of the clusters. Suppose $q'$ is the topmost layer accessed by any transaction where the diameter of the cluster on that layer is at most $d_{q'}$.

Consider the destination shard $S_j$, and we had only subtransactions from one leader shard of cluster layer $q$ where the distance between the transaction and its accessing shard is at most $d_q$, and it has maximum competitive ratio denoted by $\tau_q=O(d_q\cdot \min\{k, \sqrt{s}\})$ (from Theorem~\ref{theorem:stateless-general-graph-competative-ratio}) than any other cluster. 
% Then from Theorem~\ref{theorem:stateless-general-graph-competative-ratio} the time to process transactions is $O(ld_q\cdot min\{k,\sqrt{s}\}) = c_1ld_q\cdot min\{k,\sqrt{s}\}$ for some positive constant $c_1$. 
But now the destination shard $S_j$ needs to process subtransactions from all layers $0,\ldots, q'$ and from sublayers $0, \ldots, H_2-1$, and those transactions are processed according to their assigned order.

As discussed in Section~\ref{sec:shard-clustering}, a cluster at layer $q$ has a diameter at most $O(2^q \log s)$. Thus $d_q = O(2^q\log s) = c2^q\log s$, for some positive constant $c$. This implies $\sum_{q=0}^{q'}d_q \leq 2d_{q'}$.
%  Let $C$ be a cluster that takes the maximum time $\tau$  to process transactions compared to any other cluster in the hierarchy. Then, we can write
%     $$\tau \leq \mathcal{CR}_{\mathcal{SL}} \cdot \tau^*$$ where
%     $\tau^*$ is the optimal scheduling time for $C$ and if we have only one cluster, $\mathcal{CR}_{\mathcal{SL}}$ is the competitive ratio of the single leader
%     scheduler($\mathcal{SL}$).
%     Now, we consider all the $H_1$ layers and $H_2$ sublayers.
% The algorithm will only use transactions up to a layer $q' = O(\log d)$,
% since the clusters at that layer cover all the transactions and the respective
% objects that they access because the maximum distance between any transaction and its accessing object is at most $d$.
Thus, the competitive ratio of Algorithm \ref{alg:stateless-multi-leader-scheduler} considering transactions from all layers and sublayers at destination shard $S_j$
% at $S_q$
is at most:
\begin{equation}
\label{eqn:stateless-combined-layer-sublayer}
 \tau_{total} \leq \sum_{q=0}^{q'} \sum_{r=0}^{H_2-1} \tau_{q} \leq \sum_{q=0}^{q'} \sum_{r=0}^{H_2-1} O(d_q\cdot \min\{k, \sqrt{s}\}) \leq O(d_{q'}H_2\cdot min\{k,\sqrt{s}\}) \ .
\end{equation}

We can replace $H_2 = O(\log s)$ and $d_{q'} = O(d\log s)$ (see Section~\ref{sec:shard-clustering}), then Equation~\ref{eqn:stateless-combined-layer-sublayer} becomes:
$$ O(d\log^2 s \cdot min\{k,\sqrt{s}\}) \ .$$

\end{proof}

%% file: stateful/single-leader.tex
\section{Stateful Scheduler}
\label{sec:stateful-scheduler}
In this scheduler model, the leader shard gathers all of the transactions and the current states of the accessing accounts and pre-commits the transactions at the leader. After that, the leader creates the pre-committed subtransactions batch 
%to the respective shard 
and sends that batch to the respective destination shard, where each destination shard reaches a consensus on the received subtransaction order and adds it to their local blockchain. We provide two stateful scheduling algorithms, one with a single leader and the other with multiple leaders.

\subsection{Stateful Single-Leader Scheduler}
\label{sec:single-leader-scheduler}
We present and analyze the {\em stateful single-leader scheduler}, where one of the shards is considered as the leader $S_\ell$, which is responsible for scheduling and processing all the transactions.

When a new transaction $T_i$ is generated at its home shard $S_i$, $S_i$ sends $T_i$ to the leader shard $S_\ell$. Upon receipt, $S_\ell$ appends $T_i$ to its local pending queue $PQ_\ell$. Scheduling event is triggered periodically, either every $4\lambda$ time units or upon processing transactions associated with $\lambda$ distinct colors. Here, $\lambda$ denotes the worst-case communication delay between any two shards, which is at most the diameter of the shard communication graph $G_s$. The $4\lambda$ bound accounts for the communication delays involved in acquiring state information from remote shards and completing the pre-commitment phase and sending the pre-committing batch to the destination shard.

\begin{algorithm*}[t]
\small
%\smaller[1]
\caption{{\sc Steteful Single Leader Scheduler}}
\label{alg:stateful-single-leader-scheduler}

$S_{\ell}$: Leader shard;
$PQ_{\ell}$: Pending txns queue in leader shard\; 
$\T_{\ell}$: Set of scheduled txns maintained by leader; $G_{\T_{\ell}}$: Conflict txn graph on $\T_{\ell}$\;
$\lambda$: worst communication delay between any two shards due to partial-synchrony\;

$PrecommitSubTxnBatch(S_j)$: Precommitted subtransactions batch for shard $S_j$\;

\BlankLine
% \textbf{\textcolor{black}{Step 1}}\\
\SetKwBlock{HomeShardEvent}{\normalfont {\bf Upon generation of a new txn $T_{i}$ at home shard $S_i$}}{}
\HomeShardEvent{
     $S_i$ sends $T_{i}$ to the leader shard $S_\ell$\;\label{alg3:home-shard-send-txn-to-leader}
}

\BlankLine
% \textbf{\textcolor{black}{Event on Leader Shard $S_\ell$}}\\
\SetKwBlock{ReceiveTxnAtLeader}{\normalfont {\bf Upon receiving new txn $T_{i}$ at leader shard $S_\ell$}}{}
\ReceiveTxnAtLeader{
    $S_\ell$ appends $T_{i}$ to $PQ_{\ell}$\;

    \If{$S_\ell$ waits for $4\lambda$ time unit or $S_\ell$ proceed $\lambda$ number of scheduled colors}{\label{alg3:leader-schedule}
        \tcp{Trigger scheduling event}
        % $LastSchEventTimeUnit \leftarrow CurrentTimeUnit$\;

        Move txns from $PQ_{\ell}$ to $\T_{\ell}$;
        Identify set of accessed accounts $\mathcal{O}_v$ by txns in $\T_{\ell}$\;
        \If{Current state of any account $O_j \in \mathcal{O}_v$ is not locally available at $S_\ell$}{
            For each $O_j$,
                determine the responsible shard $S_j$ and create
                 request batch for each $S_j$\;
            
            Send batched account state request to each destination shard $S_j$\;
        }
        \Else{
          $S_\ell$ has all accounts state, so trigger internal state-ready event (see below)\;
        }
    }
}

\BlankLine
\SetKwBlock{OnStateRequest}{\normalfont {\bf Upon receiving a batched account state request at destination shard $S_j$}}{}
\OnStateRequest{
    Respond to leader shard $S_\ell$ with current states of all requested accounts\;
}

\BlankLine
\SetKwBlock{OnStateReady}{\normalfont {\bf Upon receiving account states from each $S_j$, or already available locally at $S_\ell$}}{}
\OnStateReady{
    $S_\ell$ extend txn conflict graph $G_{\T_{\ell}}$ with new txns in $\T_{\ell}$ and runs incremental greedy vertex coloring algorithm on $G_{\T_{\ell}}$  using $\zeta$ colors without altering already scheduled old txns\; \label{alg3:graph-coloring}
   
    \ForEach{color $\zeta_c\in\zeta$}{\label{alg3:pre-commit-loop-start}
        Pre-commit or abort txns $T_i\in \zeta_c$\ by checking txn condition and account state\;
            If $T_i$ is pre-committed, split $T_i$ into subTxns and create (append) pre-committed subtxns batch order $PrecommitSubTxnBatch(S_j)$ for each destination shard $S_j$\;
            
            Remove $T_i$ from $\T_{\ell}$ and $G_{\T_{\ell}}$. Send the outcome(committed/aborted) to home shard of $T_i$\;
            \tcp{Track processed color}
       
    \If{processed $\lambda$ number of colors}{\textbf{break\;}}
    } \label{alg3:pre-commit-loop-end}
     $S_\ell$ sends $PrecommitSubTxnBatch(S_j)$ to corresponding destination shard $S_j$ parallelly and start to process next batch\;\label{alg3:send-to-destinaiton}

}

\BlankLine

\SetKwBlock{OnPrecommitBatch}{\normalfont {\bf Upon receiving precommitted batch $PrecommitSubTxnBatch(S_j)$ at each $S_j$}}{}
\OnPrecommitBatch{
    Reach consensus on $PrecommitSubTxnBatch(S_j)$ and append batch to the local blockchain\; \label{alg3:reach-consensus-on-destinaiton}
    % Optionally: acknowledge to $S_\ell$ for final confirmation (not shown)\;
}
\end{algorithm*}

When the scheduling event is triggered, the leader shard moves its pending transactions from $PQ_\ell$ into the scheduling transaction set $\T_\ell$ and identifies the set of accounts $\mathcal{O}_v$ accessed by transactions which are in $\ T_\ell$. If any account state $O_j \in \mathcal{O}_v$ is not locally available at $S_\ell$, it determines the responsible destination shard $S_j$ for each such account, and sends batched account state requests to the corresponding shards. If all required states are already available in $S_\ell$, an internal {\sc State-Ready} event is triggered immediately.

Upon receiving a state request, each destination shard $S_j$ responds with the current state of the requested accounts (e.g., balances). Then, once all necessary account states are collected at $S_\ell$, it extends the conflict graph $G_{\T_\ell}$ by incorporating the new transactions in $\T_\ell$ Then the leader shard $S_\ell$  runs the incremental greedy vertex coloring algorithm~\cite{busch2022dynamic} on $G_{\T_\ell}$ and assigns at most $\zeta$ colors without altering the coloring of previously scheduled old transactions.

The leader then iteratively processes transactions color by color. For each color group $\zeta_c$, $S_\ell$ verifies transaction conditions (e.g., sufficient balance) using the up-to-date account state it gathers. Transactions that are valid and conditions are satisfied are \emph{pre-committed}, while invalid ones are aborted. Then $S_\ell$ splits each pre-committed transaction $T_i$ into subtransactions $T_{i,j}$ based on its accessed shards. These subtransactions are appended to a corresponding pre-commit batch $PrecommitSubTxnBatch(S_j)$ for each destination shard $S_j$. After processing a transaction, it is removed from $\T_\ell$ and $G_{\T_\ell}$, and the outcome (committed or aborted) is reported back to the transaction's home shard $S_i$ to initiate the next transaction.

The pre-commitment phase terminates once $\lambda$ colors are processed, after which $S_\ell$ dispatches all $PrecommitSubTxnBatch(S_j)$ batches to their respective destination shards in parallel. Each destination shard $S_j$ then reaches consensus on the order of subtransactions in the received batch and appends them to its local blockchain. The leader shard $S_\ell$, then waits and proceeds to the next scheduling batch.

% The Correctness Analysis of Single Leader Scheduler (Algorithm~\ref{alg:stateful-single-leader-scheduler}) is deferred to Appendix~\ref{app:safety-and-liveness-analysis-alg:stateful-single-leader-scheduler}.

\subsubsection{Correctness Analysis of Stateful Single-Leader Scheduler (Algorithm~\ref{alg:stateful-single-leader-scheduler})}
\label{app:safety-and-liveness-analysis-alg:stateful-single-leader-scheduler}

% \subsubsection{Correctness Analysis of Single Leader Scheduler (Algorithm~\ref{alg:stateful-single-leader-scheduler})} 

% Our proposed scheduling algorithm works on a partial-synchronous communication model, and each shard runs a consensus algorithm(e.g., PBFT~\cite{PBFT}) to reach a consensus among the nodes for correctness. In this section, for the sake of analysis only, we consider the synchronous communication mode.

\begin{lemma}[Safety]
\label{lemma:stateful-single-leader-safety}
   If two transactions conflict with each other in Algorithm~\ref{alg:stateful-single-leader-scheduler}, then they commit in different time slots. Furthermore, the local chains produced by Algorithm~\ref{alg:stateful-single-leader-scheduler} ensure global blockchain serialization.
\end{lemma}

\begin{proof}
We prove this by induction (analyzing) on the execution of Algorithm~\ref{alg:stateful-single-leader-scheduler}. The leader shard $S_\ell$ constructs the transaction conflict graph and applies the incremental greedy vertex coloring algorithm~\cite{busch2022dynamic} (Line~\ref{alg3:graph-coloring}). This algorithm guarantees that conflicting transactions receive different colors. Then the leader shard $S_\ell$ pre-commits the transactions according to the color they received (Lines~\ref{alg3:pre-commit-loop-start}–\ref{alg3:pre-commit-loop-end}). Thus, the conflicting transactions are committed in different time slots. Additionally, the leader creates batches of pre-committed subtransactions and sends them to the corresponding destination shards (Line~\ref{alg3:send-to-destinaiton}). Each destination shard reaches consensus on the received batch and appends it to its local blockchain (Line~\ref{alg3:reach-consensus-on-destinaiton}). Since the commit order is determined by the leader and this order is preserved across all destination shards, this ensures global serializability.
\end{proof}

\begin{lemma}[Liveness]
\label{lemma:stateful-single-leader-liveness}
    Algorithm~\ref{alg:stateful-single-leader-scheduler} guarantees that every generated transaction will eventually be either committed or aborted.
\end{lemma}
\begin{proof}
We prove by induction that every transaction progresses through the system without indefinite delay. When a transaction $T_i$ is generated at its home shard, it is forwarded to the leader shard $S_\ell$ (Line~\ref{alg3:home-shard-send-txn-to-leader}). The leader maintains a queue $PQ_\ell$ for pending transactions and periodically moves transactions from $PQ_\ell$ to the scheduling set $\T_\ell$ either after waiting for $4\lambda$ time units or after processing $\lambda$ colors (Line~\ref{alg3:leader-schedule}). Due to this bounded waiting and the assumption of partial synchrony, each transaction will eventually be scheduled.

Once scheduled, $T_i$ is added to the transaction conflict graph $G_{\T_\ell}$, and the leader runs an incremental greedy coloring algorithm. The algorithm ensures that new transactions are assigned colors without modifying previously scheduled old transactions (Line~\ref{alg3:graph-coloring}). The leader then pre-commits transactions color-by-color, and after processing $\lambda$ colors, it starts the next scheduling batch (Lines~\ref{alg3:pre-commit-loop-start}–\ref{alg3:pre-commit-loop-end}).
% Because new transactions are added only after old transactions have been processed, and because each round guarantees progress of $\lambda$ colors, all transactions are eventually scheduled for pre-commit.
Each pre-committed transaction is either committed (if conditions are met) or aborted (if conditions fail), and this decision is sent to the home shard. Thus, the algorithm guarantees that every transaction will eventually reach a decision (commit or abort), ensuring liveness.
\end{proof}
 
% \begin{lemma}[Liveness]
% \label{lemma:stateful-single-leader-liveness}
%     Algorithm~\ref{alg:stateful-single-leader-scheduler} guarantees that every generated transaction will eventually be either be committed or aborted.
% \end{lemma}
% \begin{proof}
%    We prove this by analyzing the Algorithm~\ref{alg:stateful-single-leader-scheduler}, every generated transaction is send to the leader and the leader shard fairly constructs the conflict graph, assigns the color, and determines the schedule of the transaction without affecting already scheduled old transaciton (Line~\ref{alg3:graph-coloring}). This ensures that new transactions will not affect the processing of old transactions. Moreover, each scheduled transaction commits according to the color it gets; this ensures the continuous commit progress of the transactions.  
% \end{proof}

\begin{corollary}
    From Lemma~\ref{lemma:stateful-single-leader-safety} and Lemma~\ref{lemma:stateful-single-leader-liveness}, Algorithm~\ref{alg:stateful-single-leader-scheduler} ensures the safety and liveness of the transactions.
\end{corollary}

\subsubsection{Performance Analysis of Single-Leader Scheduler (Algorithm~\ref{alg:stateful-single-leader-scheduler})} 
In the following, we analyze the time unit required to process transactions by Algorithm \ref{alg:stateful-single-leader-scheduler}. 
We focus on the special case where the maximum distance between
the transactions, their accessing objects, and the leader is at most $d$, and at least one transaction accesses objects at a distance $\Omega(d)$.
This special case is useful for the analysis of the multi-leader case.
We are focusing on the time period after the leader shard has determined the schedule for the transactions. This is because the scheduling and committing steps are executed in parallel.

\begin{theorem}
\label{theorem:stateful-upper-bound-general-graph-competative-ratio}
    [General Graph] In the General graph, where the transactions, their accessing objects, and the leader are at most $d$ distance away from each other, and at least one transaction is
    $\Omega(d)$ distance from the accessing shards, Algorithm~\ref{alg:stateful-single-leader-scheduler} has $O(\min\{k, \sqrt{s}\})$ competitive ratio.
\end{theorem}

\begin{proof} This proof follows the same arguments discussed in the proof of Theorem~\ref{theorem:stateless-general-graph-competative-ratio}.  Consider a set of transactions $\T$ generated at or before time $t$ that are still pending (neither committed nor aborted) at time $t$.
%Let $\T' \subseteq \T$ denote the transactions that will be colored at the leader,
%while the remaining $\T'' = \T \setminus \T'$
%have already received a color.
%The choice of color for each $\T'$
Let $G_{\T}$ denote the conflict graph for $\T$,
where two transactions conflict if they have a common destination shard.
Let $l_i$ denote the number of transactions in $\T$ that use objects in shard $S_i$.
Let $l = \max l_i$. Moreover, from the definition of $d$, at least one transaction is $d$ distance away from the destination shard or leader. So we have that $\Omega(l+d)$ is a lower bound on the time 
that it takes to finalize (commit or abort) the transactions in $\T$, since at least $l$ subtransactions need to serialize in a destination shard, and at least one transaction is $d$ distance away. 

First, consider the case where $k \leq \sqrt s$.
We have that each transaction conflicts with at most $k l$
other transactions. Hence $G_{\T}$ can be colored with
at most $k l + 1$ colors.

Algorithm~\ref{alg:stateful-single-leader-scheduler} schedules and commits transactions in batches. For each batch, the leader shard performs the following steps: first, it gathers the state of accessed accounts, takes at most $2d$ time units (request and receive each takes at most $d$ time units). After pre-committing, the leader sends the pre-commit batch to destination shards, which takes $d$ time units.
Additionally, destination shards reach consensus on the received batch within $1$ time unit. Hence, the total delay per batch is at most $3d + 1$.

Since the algorithm uses at most $kl + 1$ colors (batches), the total finalization time is at most:
$ kl + 1 + 3d + 2 = O(kl + d)$. 

%Therefore, the competitive ratio is:
%$ \frac{O(kl + d)}{l + d} = O(k)$.

Next, consider the case $k > \sqrt s$. Following the same reasoning above and from Theorem~\ref{theorem:stateless-general-graph-competative-ratio}, we get $O(l \sqrt{s} + d)$
time to finalize the transactions $\T$.

%competitive ratio. 
%Thus, the schedule is an  $O(\min\{k,\sqrt{s}\})$ approximation.

Overall, Algorithm~\ref{alg:stateful-single-leader-scheduler} requires
$O(l\cdot \min\{k,\sqrt{s}\}+d)$ time units to finalize the transactions.
Since $\Omega(l + d)$ is a lower bound,
we have that the approximation factor of the schedule for $\T$ is $O(\min\{k,\sqrt{s}\})$.

% We can write $\T' = A \cup B$,
% where $A$ are the transactions which access at most $\sqrt s$
% destination shards, while $B$ are the transactions which access more than $\sqrt s$ destination shards.
% Each transaction in $A$ conflicts with at most $l \sqrt s$ other transactions. Hence, the transactions in $A$ need at most $l \sqrt s + 1$ distinct colors.
% The transactions in $B$ can be serialized, requiring at most $|B|$ distinct colors.
% Hence, the conflict graph $G_T$ can be colored with at most
% $l \sqrt s + 1 + |B|$ colors, which implies a schedule of 
% length $O(d (l \sqrt s +  |B|))$ steps to finalize the transactions $\T$.
% Since each transaction in $B$ accesses more than $\sqrt s$ shards,
% there is a shard accessed by more than $(|B| \sqrt s) / s = |B| / \sqrt s$ transactions. 
% Thus, $l > |B| / \sqrt s$.
% Hence,
% for transactions $\T$,
% the approximation of their finalization time is 
% $O(d(l \sqrt s +  |B|)/ l) = O(d \sqrt s + d |B| / l) = O(d \sqrt s +  d \sqrt s) = O(d \sqrt s)$.

%Therefore, combining the approximations for the cases $k \leq \sqrt s$ and $k > \sqrt s$,
%we have that the combined approximation for the finalization time for $\T$ is $O(\min\{k,\sqrt s\})$.
Since $t$ is chosen arbitrarily,
we have that the competitive ratio of Algorithm~\ref{alg:stateful-single-leader-scheduler} is $O(\min\{k,\sqrt s\})$.
\end{proof}

%% file: stateful/multi-leader.tex
\subsection{Stateful Multi-Leader Scheduler}
\label{sec:stateful-multi-leader-scheduler}
We present a {\em stateful multi-leader scheduler} in which multiple leader shards are responsible for scheduling and processing transactions. 
In the single-leader algorithm, the value $d$ includes the distance to the leader, but in the multi-leader, $d$ does not include the relative distance to the leader.
This allows the multi-leader algorithm to capture better the locality of transactions, allowing for shorter distance messages between the involved home and destination shards.

\begin{algorithm*}[t]
%\smaller[1]
\small
%\smaller[1]
\caption{\sc Stateful Multi-Leader Scheduler}
\label{alg:stateful-multi-leader-scheduler}

Each shard knows the hierarchical cluster decomposition of $G_s$\;
Each cluster $C(q,r)$ has: leader shard $S_{\ell}$, txn queue $PQ_\ell$, scheduled txns $\T_\ell$, conflict graph $G_{\T_\ell}$\;
$\lambda_C$: worst communication delay between any two shards in cluster $C$ due to partial-synchrony\;

{\tt scheduleControl}: Boolean flag indicating whether the cluster currently holds scheduling control\;

\BlankLine

\SetKwBlock{OnTxnGenerated}{\bf Upon generation of new txn $T_i$ at home shard $S_i$}{}
\OnTxnGenerated{
    $S_i$ determines the lowest cluster $C(q, r)$ which includes $T_i$ and its accessing shards. Then  $S_i$ sends $T_i$ to leader shard $S_\ell$ of $C(q, r)$\;
}

\BlankLine

\SetKwBlock{OnTxnReceivedAtLeader}{\bf Upon receiving txn(s) $T_i$ at leader shard $S_\ell$ of $C(q, r)$}{}
\OnTxnReceivedAtLeader{
    $S_\ell$ appends $T_i$ to its pending transactions queue $PQ_\ell$\;
% }

% \BlankLine

% \SetKwBlock{OnScheduleCheck}{\bf Periodic {\sc OnScheduleCheck} at leader $S_\ell$ of $C(q, r)$}{}
% \OnScheduleCheck{
    \If{$S_\ell$ waits for $4\lambda_C$ time unit or $S_\ell$ proceed $\lambda_C$ number of previous scheduled colors}{
        \If{${\tt scheduleControl} == \text{True}$}{

            \tcp{Invoke single-leader scheduling logic}
            Run Single-Leader Scheduler (Algorithm~\ref{alg:stateful-single-leader-scheduler}) with $(PQ_\ell, \T_\ell, G_{\T_\ell}, \lambda_C)$\;
            \tcp{If Algorithm~\ref{alg:stateful-single-leader-scheduler} break after process $\lambda_C$ number of scheduled colors then check and do following:}
            \If{parent cluster $C'$ requests control}{
                Send {\tt scheduleControl} to the parent and set
                ${\tt scheduleControl} \gets \text{False}$\;
            }

            \ElseIf{children clusters $C''$ request control}{
                Send {\tt scheduleControl} to children and set
                ${\tt scheduleControl} \gets \text{False}$\;
            }
            \ElseIf{$C(q,r)$ doesn't have remaining transactions to schedule}{
                Send {\tt scheduleControl} down to children and set
                ${\tt scheduleControl} \gets \text{False}$\;
            }

        } \Else {
            % \tcp{Try to reacquire control}
            Send request to current {\tt scheduleControl} holder  (e.g., child or parent cluster);
        }
    }
}

\BlankLine

\SetKwBlock{OnControlReceived}{\bf Upon receiving {\tt scheduleControl} at leader $S_\ell$ of $C(q, r)$}{}
\OnControlReceived{
\If{$S_\ell$ previously requested $scheduledControl$ to process its txns}{
 Set ${\tt scheduleControl} \gets \text{True}$ and trigger internal event (see above on line 9-12)\;
}
\Else{
Send {\tt scheduleControl} to parent or child clusters according to the request it gets\;
}
}
\end{algorithm*}

The system assumes a hierarchical cluster decomposition~\cite{gupta2006oblivious} of the shard graph $G_s$, which is globally known to all shards. Each cluster $C(q, r)$ in the hierarchy is associated with a leader shard $S_\ell$, a pending transaction queue $PQ_\ell$, a scheduled transaction set $\T_\ell$, and a transaction conflict graph $G_{\T_\ell}$. The parameter $\lambda_C$ denotes the worst-case communication delay between any two shards within the cluster $C$, which arises from the assumption of a partially synchronous communication model.

% \subsubsection{Scheduling Algorithm}

In multi-leader scheduling Algorithm~\ref{alg:stateful-multi-leader-scheduler}, when a transaction $T_i$ is generated at its home shard $S_i$, the shard identifies the lowest cluster $C(q, r)$ that contains all the shards accessed by $T_i$, and then forwards $T_i$ to the leader shard $S_\ell$ of of cluster $C$. The leader shard $S_\ell$ appends the received transaction to its pending queue $PQ_\ell$. Periodically, the leader checks if either $4\lambda_C$ time units have elapsed since the last scheduling event or if $\lambda_C$ colors of scheduled transactions have been processed by $S_\ell$. If either condition is met and the leader holds the {\tt scheduleControl}, it invokes the single-leader scheduler (Algorithm~\ref{alg:stateful-single-leader-scheduler}) on its local structures $(PQ_\ell, \T_\ell, G_{\T_\ell}, \lambda_C)$ to process transactions.

The scheduling control, denoted by the boolean flag \texttt{scheduleControl}, determines which cluster can perform scheduling operations at a given time unit. The control flows hierarchically between parent and child clusters. A parent cluster of $C$ is any cluster at a higher level in the hierarchy (with height $(q', r') > (q, r)$) that shares at least one shard with $C$. Similarly, a child cluster of $C$ is a lower-level cluster (with height $(q'', r'') < (q, r)$) that overlaps with $C$. Clusters may have multiple parents and children. If $C$ is at the bottom-most level (height $(0,0)$), initially it has {\tt scheduleControl}. Otherwise, it must request control from all its children. Once all children respond the {\tt scheduleControl}, the leader $S_\ell$ sets \texttt{scheduleControl} to true and proceeds with the scheduling.

After executing the single-leader scheduler, if the parent cluster $C'$ requests control, the leader transfers \texttt{scheduleControl} to the parent and sets it to false locally. If instead a child cluster $C''$ has made a request, the control is passed down to the child. If there are no remaining transactions to process, the control is also passed downward to allow lower-level clusters to schedule pending transactions. If the leader does not have {\tt scheduleControl} when scheduling should occur, it sends a control request to the current holder (parent or child). Additionally, if $C$ receives a control request from a parent $C'$ while not holding control, it forwards the request to its children. Once all children respond positively, it passes control up to $C'$. This hierarchical and event-driven mechanism ensures coordinated and conflict-free scheduling across multiple levels of the cluster hierarchy.

\subsubsection{Correctness Analysis of Stateful Multi-Leader Scheduler (Algorithm~\ref{alg:stateful-multi-leader-scheduler})}
\label{app:safety-and-liveness-analysis-stateful-multi-leader-scheduler}
\begin{lemma}[Safety]
\label{lemma:multi-leader-safety}
   If two transactions conflict with each other in Algorithm~\ref{alg:stateful-multi-leader-scheduler}, then they will commit in different time slots, and the local chain produced by Algorithm~\ref{alg:stateful-multi-leader-scheduler} ensures blockchain serialization.
\end{lemma}
As each cluster $C$ of Algorithm~\ref{alg:stateful-multi-leader-scheduler} involves the Algorithm~\ref{alg:stateful-single-leader-scheduler}, the proof follows the same reasoning as Lemma~\ref{lemma:stateful-single-leader-safety}.

\begin{lemma}[Liveness]
\label{lemma:multi-leader-liveness}
    Algorithm~\ref{alg:stateful-multi-leader-scheduler} guarantees that every generated transaction will eventually be committed or aborted.
\end{lemma}
As each cluster $C$ of Algorithm~\ref{alg:stateful-multi-leader-scheduler} involves the Algorithm~\ref{alg:stateful-single-leader-scheduler}, the proof follows the same reasoning as Lemma~\ref{lemma:stateful-single-leader-liveness}.

\begin{corollary}
    From Lemma~\ref{lemma:multi-leader-safety} and Lemma~\ref{lemma:multi-leader-liveness}, Algorithm~\ref{alg:stateful-multi-leader-scheduler} ensures the safety and liveness of the transactions.
\end{corollary}

\subsubsection{Performance Analysis of Multi-Leader Scheduler (Algorithm~\ref{alg:stateful-multi-leader-scheduler})}
The multi-leader scheduler is the extended version of the single-leader scheduler (Algorithm~\ref{alg:stateful-single-leader-scheduler}) while introducing an additional overhead cost due to its shard (hierarchical) clustering structure and comes from the layers and sublayers of the clusters.

\begin{theorem}
\label{lemma:stateful-upper-bound-multi-leader-scheduler}
% In Multi-leader scheduler (Algorithm \ref{alg:stateful-multi-leader-scheduler}), if the new transaction is generated at time unit $t$, then the time to process (commit/abort) all of the transactions at time moment of $t$ is at most $O(l \log s \cdot \min\{k, \sqrt{s}\}+d\log^2 s)$.
In Multi-leader scheduler (Algorithm \ref{alg:stateful-multi-leader-scheduler}), 
 where the transactions and their accessing objects are at most $d$ distance away from each other, Algorithm~\ref{alg:stateful-multi-leader-scheduler} has $O(\log s\cdot \min\{k, \sqrt{s}\}+\log^2 s)$ competitive ratio.
\end{theorem}
\begin{proof}
% This proof is similar to the proof of Theorem~\ref{theorem:stateless-upper-bound-multi-leader-scheduler}.
% ; however, instead of using Theorem~\ref{theorem:stateless-upper-bound-general-graph} to proof Theorem~\ref{theorem:stateless-upper-bound-multi-leader-scheduler}, here we use Lemma~\ref{theorem:stateful-upper-bound-general-graph}.
Similar to Theorem~\ref{theorem:stateless-upper-bound-multi-leader-scheduler}, consider the destination shard $S_j$, as discussed in the proof of Theorem~\ref{theorem:stateful-upper-bound-general-graph-competative-ratio}, if we had only subtransactions from one leader shard of cluster layer $q$ where the distance between the transaction and its accessing shard is at most $d_q$, then the time to process transactions is $O(l \cdot \min\{k, \sqrt{s} \} + d_q)$ or equivalently at most $c_1(l \cdot \min\{k, \sqrt{s}\} + d_q)$ time for some positive constant $c_1$. Suppose $q'$ is the maximum layer accessed by any transaction where the diameter of the cluster on that layer is at most $d_{q'}$.  Then the destination shard $S_j$ needs to process subtransactions from all layers $0,\ldots, q'$ and from sublayers $0, \ldots, H_2-1$, and those transactions are processed according to their assigned order.

As discussed in Section~\ref{sec:shard-clustering}, a cluster at layer $q$ has a diameter at most $O(2^q \log s)$. Thus $d_q = O(2^q\log s) = c2^q\log s$, for some positive constant $c$. This implies $\sum_{q=0}^{q'}d_q \leq 2d_{q'}$.
Thus, the total time unit required by Algorithm \ref{alg:stateful-multi-leader-scheduler} to process all the transactions from all layer and sublayers at destination shard $S_j$
% at $S_q$
is at most:
\begin{equation}
\label{eqn:stateful-combined-layer-sublayer}
 \tau_{total} \leq \sum_{q=0}^{q'} \sum_{r=0}^{H_2-1}  c_1(l \cdot \min\{k, \sqrt{s} \} + d_q) \leq c_1lH_2\cdot min\{k,\sqrt{s}\} + 2c_1d_{q'}H_2 \ .
\end{equation}

We can replace $H_2 = c_2\log s$ as we have $O(\log s)$ sublayers (see Section~\ref{sec:shard-clustering}) and $d_{q'} = c_3 d\log s$, where $c_2$ and $c_3$ are some positive constants, then Equation~\ref{eqn:stateful-combined-layer-sublayer} becomes:
$$c_1l\cdot c_2\log s \cdot min\{k,\sqrt{s}\} + 2c_1\cdot c_3d\log s \cdot c_2 \log s =>$$
$$O(l\log s \cdot min\{k,\sqrt{s}\}+d\log^2s) \ .$$

As discussed in Theorem~\ref{theorem:stateful-upper-bound-general-graph-competative-ratio}, $\Omega(l+d)$ is a lower bound. Thus, we have that the competitive ratio of Algorithm~\ref{alg:stateful-multi-leader-scheduler} as $O(\log s\cdot \min\{k, \sqrt{s}\}+\log^2 s)$.
 
\end{proof}

% \begin{theorem}
% \label{theorem:stateful-multi-leader-general-graph-competative-ratio}
%      In the General graph, where the transaction and its accessing objects are $d$ distance away, Algorithm~\ref{alg:stateful-multi-leader-scheduler} has $O(\log s\cdot \min\{k, \sqrt{s}\}+\log^2 s)$ competitive ratio.
% \end{theorem}
% \begin{proof}
%     From Lemma~\ref{lemma:stateful-upper-bound-multi-leader-scheduler}, if shards are connected in a general graph, then Algorithm~\ref{alg:stateful-multi-leader-scheduler} takes at most $O(l \log s \cdot \min\{k, \sqrt{s}\}+d\log^2 s)$ time units to process the transactions. Similarly, from Lemma~\ref{lemma:lower-bound-single-leader}, $\Omega(l+d)$ is the lower bound. This gives us $O(\log s\cdot \min\{k, \sqrt{s}\}+\log^2 s)$ competitive ratio.
% \end{proof}

%% file: conclusion.tex
\section{Conclusion}
\label{sec:conclusion}

We presented efficient scheduling algorithms for processing dynamic transactions in blockchain sharding systems. Our proposed framework operates under a partially synchronous communication model, which realistically captures the behavior of many real-world blockchain environments.
We introduced both stateless and stateful scheduling models, each of which includes single-leader and multi-leader algorithms for transaction scheduling and processing. For these algorithms, we provided competitive ratios relative to an optimal scheduler and established both upper and lower bounds on the scheduling delay. To the best of our knowledge, this is the first provably efficient dynamic transaction scheduling framework tailored for blockchain sharding.

For future work, we plan to explore efficient inter-shard communication mechanisms, particularly under conditions of network congestion where communication links have bounded capacity. We also aim to conduct extensive simulations and real-world experiments to evaluate the practical performance of our proposed protocols.